\documentclass[journal,twoside,web]{ieeecolor}

\usepackage{generic}
\usepackage{amsmath,amssymb,amsfonts,mathtools}
\usepackage{graphicx}
\let\labelindent\relax
\usepackage{enumitem}
\usepackage{stmaryrd}

\DeclareMathOperator*{\argmin}{argmin}
\usepackage{marginnote}

\definecolor{myred}{rgb}{0.8, 0.0, 0.0}

\definecolor{myblue}{rgb}{0.0, 0.0, 0.8}

\DeclareMathOperator*{\esup}{ess\,sup}

\newcommand{\virg}[1]{``#1''}
\newcommand{\mzr}[1]{\lceil#1\rceil_{\text{M}}}
\newcommand{\abs}[1]{{\left\vert #1 \right\vert}}
\newcommand{\norm}[1]{{\left\vert\kern-0.25ex\left\vert #1\right\vert\kern-0.25ex\right\vert}}
\newcommand{\nnorm}[1]{{\left\vert\kern-0.25ex\left\vert-0.25ex\left\vert #1\right\vert\kern-0.25ex\right\vert-0.25ex\right\vert}}
\newcommand{\norminf}[1]{{\left\vert\kern-0.25ex\left\vert #1\right\vert\kern-0.25ex\right\vert}_{\infty}}
\newcommand{\weakp}[1]{{\left[\kern-0.25ex\left[ #1\right]\kern-0.25ex\right]}}

\newcommand{\rea }{\mathbb{R}}
\newcommand{\nat}{\mathbb{N}}

\newcommand{\ub}{\boldsymbol{u}}
\newcommand{\xb}{\boldsymbol{x}}

\newcommand{\pb}{\boldsymbol{p}}

\newcommand{\bb}{\boldsymbol{b}}
\newcommand{\yb}{\boldsymbol{y}}
\newcommand{\vb}{\boldsymbol{v}}

\newcommand{\etab}{\boldsymbol{\eta}}
\newcommand{\Xc}{\mathcal{X}}

\newcommand{\Ls}{\mathsf{L}}

\newcommand{\Fs}{\mathsf{F}}
\newcommand{\Ss}{\mathsf{S}}

\newcommand{\Ts}{\mathsf{T}}
\newcommand{\ts}{\mathsf{t}}
\newcommand{\Is}{\mathsf{Id}}

\newcommand{\As}{\mathsf{A}}

\newcommand{\Bc}{\mathcal{B}}
\newcommand{\Hc}{\mathcal{H}}

\newcommand{\Gc}{\mathcal{G}}
\newcommand{\Nc}{\mathcal{N}}

\newcommand{\bone}{\boldsymbol{1}}
\newcommand{\bzero}{\boldsymbol{0}}
\newcommand{\fix}{\mathrm{fix}}
\newcommand{\zer}{\mathrm{zer}}

\newcommand{\LReLU}{\mathrm{LReLU}}
\newcommand{\diag}{\mathrm{diag}}

\newcommand{\Lip}{\mathsf{Lip}}

\usepackage{ntheorem,lipsum}
\theoremheaderfont{\hspace*{\parindent}\bf}
\theoremseparator{.}
\newtheorem{assum}{Assumption}
\newtheorem{thm}{Theorem}
\newtheorem{rem}{Remark}
\newtheorem{prop}{Proposition}
\newtheorem{cor}{Corollary}
\newtheorem{lem}{Lemma}
\newtheorem{defn}{Definition}
\newtheorem{exmp}{Example}

\usepackage[T1]{fontenc}

\usepackage{csquotes}

\makeatletter
\let\c@author\relax
\makeatother
\usepackage[citestyle=numeric-comp,
 backend=biber,
 style= ieee,
 sorting=none,
 giveninits=true,
 sortcites=true]{biblatex}
\AtEveryBibitem{
	\clearfield{issn}
	\clearfield{isbn}
	\clearfield{archivePrefix}
	\clearfield{arxivId}
	\clearfield{pmid}
	\clearfield{doi}
	\clearfield{eprint}
	\clearfield{pages}
	\clearfield{month}
	\ifentrytype{online}{}{\ifentrytype{misc}{}{\clearfield{url}}}
	\ifentrytype{book}{\clearfield{doi}}{}
}

\usepackage[shortcuts]{extdash}
\begin{document}

% On the convergence of the Krasnoselskij iteration for strictly pseudocontractive operators

\title{\huge Non-Euclidean Enriched Contraction Theory for\\ Monotone Operators and Monotone Dynamical Systems}
\author{Diego Deplano, Sergio Grammatico, Mauro Franceschelli
\thanks{Work supported by project e.INS- Ecosystem of Innovation for Next Generation Sardinia (cod. ECS 00000038) funded by the Italian Ministry for Research and Education (MUR) under the National Recovery and Resilience Plan (NRRP) - MISSION 4 COMPONENT 2, \virg{From research to business} INVESTMENT 1.5, \virg{Creation and strengthening of Ecosystems of innovation} and construction of \virg{Territorial R\& D Leaders}.}
\thanks{Diego Deplano and Mauro Franceschelli are with the DIEE, University of Cagliari, 09123, Italy. Emails: {\tt\footnotesize diego.deplano@unica.it, mauro.franceschelli@unica.it}.}
\thanks{Sergio Grammatico is with the Delft Center for Systems and Control, TU Delft, The Netherlands. Email: {\tt\footnotesize s.grammatico@tudelft.nl}.}
}

\maketitle
\begin{abstract}
We adopt an operator‐theoretic perspective to analyze a class of nonlinear fixed‐point iterations and discrete‐time dynamical systems.
Specifically, we study the Krasnoselskij iteration---at the heart of countless algorithmic schemes and underpinning the stability analysis of numerous dynamical models---by focusing on a non‑Euclidean vector space equipped with the diagonally weighted supremum norm. 

By extending the state of the art, we introduce the notion of \emph{enriched weak contractivity}, which (i) is characterized by a simple, verifiable condition for Lipschitz operators, and (ii) yields explicit bounds on the admissible step size for the Krasnoselskij iteration.
Our results relate the notion of weak contractivity with that of monotonicity of operators and dynamical systems and show its generality to design larger step sizes and improved convergence speed for broader classes of dynamical systems.

The newly developed theory is illustrated on two applications: the design of zero-finding algorithms for monotone operators and the design of nonlinear consensus dynamics in monotone multi-agent dynamical systems.
\end{abstract}

% \begin{IEEEkeywords}
% \end{IEEEkeywords}

\section{Introduction}\label{sect1}
We study the Krasnoselskij iteration:
\begin{equation}\label{eq:mainiter_first}
\xb(k+1)=(1-\theta)\xb(k)+\theta \Ts(\xb(k)),\:\: k\in\nat,
\end{equation}
where $\theta\in(0,1)$ is the step size, and $\Ts:\rea^n\rightarrow \rea^n$ is an operator. The Krasnoselskij iteration is a fundamental tool in fixed point theory, and its importance in system theory arises from its role in analyzing and designing iterative algorithms with guaranteed convergence to equilibrium points or invariant sets.
Variations of the Krasnoselskij fixed-point iteration have been adopted to design distributed algorithms for computing fixed points in networks~\cite{li2020distributed,bastianello2020asynchronous,andrade2022distributed,li2023dot,nian2023continuous}, splitting methods in distributed convex optimization~\cite{giselsson2016linear,pavel2019distributed,dall2019convergence,guo2023disa,bastianello2024robust}, aggregative game theory~\cite{grammatico2015decentralized,grammatico2017dynamic}, monotone dynamical systems~\cite{manfredi2017necessary,jafarpour2022non,Deplano23,deplano24neural,kawano2025contraction,parasnis2025perron}, and monotone operator theory~\cite{winston2020monotone,pabbaraju2021estimating,Jafarpour21}.
From a mathematical perspective, the convergence of the Krasnoselskij iteration is a fixed-point problem~\cite{Berinde07} on the operator $\Ts$, or equivalently, a zero-finding problem for the operator $\Is-\Ts$~\cite{Bauschke2017}, where $\Is$ denotes the identity operator.
For instance, consensus in nonlinear multi-agent systems is equivalent to finding a collective state in the kernel of the nonlinear Laplacian operator~\cite{bondy2008graph,Deplano18,Deplano20,bonetto2022nonlinear,bronski2014spectral,devries2018kernel,qu2025controllability}.

In many of the above-mentioned applications, monotonicity plays a central role in guaranteeing convergence and stability, which enables the use of iterative schemes, including the Krasnoselskij iteration in~\eqref{eq:mainiter_first}, to compute fixed points efficiently~\cite{winston2020monotone,pabbaraju2021estimating,Jafarpour21}.
%
% This framework underpins modern approaches in convex optimization, game theory, and signal processing, particularly through its connection with variational inequalities and subdifferential calculus.
%
On the other hand, in dynamical systems theory, monotonicity refers to systems whose trajectories preserve a partial order over time, which facilitates the design of distributed control laws in large-scale networked systems~\cite{manfredi2017necessary,jafarpour2022non,Deplano23,deplano24neural,kawano2025contraction,parasnis2025perron}.
This dual relevance of monotonicity -- both as an algebraic property of operators and as a dynamical property of systems -- has motivated for decades the study of fixed point iterations like~\eqref{eq:mainiter_first} in contexts where the operator $\Ts$ either demonstrates some form of monotonicity or induces monotonic behavior in the trajectories generated by the iteration.

The aim of this paper is to introduce \textit{enriched contraction theory} as a general framework encompassing both of these monotonicity notions.
Toward this direction, we focus on the study of the Krasnoselskij iteration in non-Euclidean vector spaces equipped with a diagonally weighted supremum norm. Understanding its convergence under more general conditions would not only extend classical results but also open up new possibilities for applications in systems, control, learning, and optimization.

\subsection{Literature review}
One of the first convergence results for Krasnoselskij iterations dates back to 1955 and is due to Krasnoselskij~\cite{Krasnoselskii55}, \cite[Theorem 6.4.1]{Agarwal09}, who proved convergence of $\xb(k)$ to a fixed point when $\Ts$ is nonexpansive and $\theta = \frac{1}{2}$ for \textit{uniformly convex} spaces~\cite[Definition 1.8]{Berinde07}.
More than ten years later, Edelstein in~\cite{Edelstein66} extended this result to $\theta \in (0,1)$ and \textit{strictly convex} spaces~\cite[Definition 1.10]{Berinde07}.
In 1976, the convergence results for the Banach-Picard iteration in~\eqref{eq:mainiter_first} in uniformly/strictly convex spaces were extended to general Banach spaces by Ishikawa~\cite[Theorem 1]{ishikawa76}, see also ~\cite[Theorem 6.4.3]{Agarwal09}.
By limiting their analysis to Hilbert spaces, Marino and Xu in~\cite{Marino07} proved that the iteration in~\eqref{eq:mainiter_first} converges also when the map $\Ts$ is $\kappa$-strictly pseudocontractive and $\theta<1-\kappa$.
Moreover, for linear maps in Hilbert spaces, it has been recently proven that $\kappa$-strictly pseudocontractivity of $\Ts$ is both necessary and sufficient for the convergence of the Krasnoselskij iteration, given $\theta<1-\kappa$~\cite[Theorem 1]{Belgioioso18}. 
Marino and Xu in~\cite[Section 3]{Marino07} also posed the currently open question: \virg{\textit{Is strict pseudocontractivity also sufficient in Banach spaces which are uniformly convex?}}.
Since then, many authors have provided different answers to this question by considering several iteration schemes and sets of assumptions~\cite{zhang2009convergence,zhang2009strong,cai2010strong,zhou2010convergence,chidume2010weak,cholamjiak2010weak,cholamjiak2011strong}.

A thorough answer to this question has been given in~\cite{deplano25withapp} in the special case of real vector non-Euclidean spaces of finite dimension $n\in\nat$ equipped with a $p$-norm such that $p\in(1,\infty)$: The Krasnoselskij iteration converges if ${\theta^{r-1}<(1-\kappa)/c_p}$ where $r=\min\{p,2\}$ and $c_p\geq 1$ is a constant that depends on $p$, whose best (smallest) value has been found by Xu in~\cite[Corollary 2]{xu1991inequalities}, and proved in~\cite[Lemmas 4-5]{deplano25withapp}.
Moreover, it has been shown that, for non-uniformly convex spaces with $p\in\{1,\infty\}$, strict pseudocontractivity is not sufficient for the convergence of the Krasnoselskij iteration.
A natural open question arises:
\begin{equation}
\parbox{0.9\linewidth}{
\centering
\virg{\textit{What is the most appropriate generalization of strict pseudocontractivity
in non-Euclidean vector spaces that are not uniformly convex?}}
}
\label{eq:open_question}
\end{equation}
% \virg{What is the most appropriate generalization of strict pseudocontractivity in non-Euclidean vector spaces that are not uniformly convex?}.
%
% In this manuscript, we propose a possible answer to this question through the concept of \textit{enriched weak contractivity}.

\subsection{Main contributions}

The first main contribution of this manuscript is to propose an answer to question \eqref{eq:open_question} by introducing the property of \textit{enriched weak contractivity}. We show that this property serves as a natural generalization of strict pseudocontractivity in Banach spaces, while they are equivalent properties in Hilbert spaces (Proposition~\ref{prop:eqKPSENE}).
Subsequently, by focusing on vector spaces equipped with a diagonally weighted norm, we prove the following original results:
\begin{itemize}
    %%% 
    \item We provide a necessary and sufficient condition for enriched weak contractivity of Lipschitz operators (Theorem~\ref{thm:inf-coenex});
    %%%
    \item We provide a bound on the maximum allowable step size ensuring the convergence of the Krasnoselskij iteration in~\eqref{eq:mainiter_first} under the assumption that $\Ts$ is enriched weakly contracting (Theorem~\ref{thm:nonlin-bene}).
\end{itemize}
Regarding the relationship with monotone operators, we prove the following original results:
\begin{itemize}
    \item A Lipschitz operator $\Ts$ is enriched weakly contractive if and only if the residual operator $\Is-\Ts$ is (strongly) monotone, where $\Is$ denotes the identity operator (Theorem~\ref{thm:mono-ene});
    %%%
    \item The Krasnoselskij iteration in~\eqref{eq:mainiter_first} converges for larger allowable range of step sizes and improved contraction factor compared with those obtained by the state of the art on monotone operators~\cite{davydov2024non} (Section~\ref{sec:monotone}).
    %%%
\end{itemize}
Next, on the relationship with monotone dynamical systems, we prove the following original results:
\begin{itemize}
    %%%
    \item For monotone dynamical systems $\xb(k+1)=\Ts(\xb(k))$, the operator $\Ts$ is enriched weakly contractive if and only if $\Ts$ is (strictly) subhomogeneous (Corollary~\ref{cor:monoKSPC});
    %%%
    \item For monotone dynamical systems whose dynamics is ruled by the Krasnoselskij iteration in~\eqref{eq:mainiter_first}, an easy-to-verify bound on the maximum step size is derived (Corollary~\ref{cor:nonlin-bene-mono}).
    %%%
\end{itemize}
Concluding the manuscript, we apply the above theoretical results on two main applications along with numerical simulations corroborating the technical findings:
\begin{itemize}
    \item Zero finding algorithms for monotone operators -- We derive sufficient conditions for the convergence of the \textit{forward step method} applied to monotone operators (Theorem~\ref{thm:fsm}) and simulate it on linear operators and nonlinear diagonal operators (Sections~\ref{sec:aff}-\ref{sec:dnl}). 
    %%%
    \item Nonlinear consensus in monotone multi-agent systems -- We derive sufficient conditions on the nonlinear local interaction rule between agents ensuring their convergence to a consensus state (Theorem~\ref{th:consensus_DT}).
\end{itemize}

\textit{Structure of the paper.} 
Section~\ref{sec:background} provides the necessary background on enriched weak contractions, along with necessary and sufficient conditions for its validity and a comparison with other important operator-theoretic properties.
Section~\ref{sec:krasno} provides sufficient conditions for the convergence of the Krasnoselskij iteration on enriched weakly contracting operators in vector spaces equipped with a diagonally weighted norm.
Sections~\ref{sec:app1}-\ref{sec:app2} discuss, respectively, the application of the theoretical results to design algorithms for computing zeros of monotone operators and distributed algorithms for reaching consensus in order-preserving multi-agent systems.

\section{Background on\\ Enriched Weak Contractivity}\label{sec:background}

\subsection{Notation and preliminaries}

The sets of real and integer numbers are
denoted by $\rea$~and~$\mathbb{Z}$, while their restrictions to nonnegative and positive values are denoted by $\rea_{\geq 0}$, $\nat$ and $\rea_+$, $\nat_+$, respectively.
Scalars $s\in\mathbb{R}$ are denoted by lowercase letters, while vectors $\vb\in\mathbb{R}^n$ by boldface bold letters. The vectors of zeros and ones of dimension $n$ are denoted by $\bzero_n$ and $\bone_n$, respectively, and the subscript $n$ is omitted if clear from the context.
Matrices ${M\in \mathbb{R}^{n\times n}}$ are denoted by uppercase letters, and $\mathbb{S}^n_{\succ 0}$ denotes the set of positive definite symmetric matrices, i.e., such that $M^\top = M$ and $\xb^\top M \xb > 0$ for all $\xb\in\rea^n$.
A matrix is said to be \virg{\textit{Metzler}} if its off-diagonal entries $m_{ij}$ (with $i\neq j$) are nonnegative. 
The nonnegative majorant $\abs{M}$ and the Metzler majorant $\mzr{M}$ of a matrix are defined entry-wise by 
\begin{equation}\label{eq:majorant}
        (\abs{M})_{ij} := \abs{m_{ij}},\qquad  (\mzr{M})_{ij} := \begin{cases}
    a_{ii} & \text{if } i=j,\\
    \abs{m_{ij}} & \text{if } i\neq j.
    \end{cases}
\end{equation}
Given a vector $\vb\in\rea^n$, the diagonal matrix whose entries are those of the vector $\vb$ is denoted by $[\vb]$, namely
\begin{equation}\label{eq:diag_weight}
    ([\vb])_{ij} := \begin{cases}
    \vb_i & \text{if } i=j,\\
    0 & \text{if } i\neq j.
    \end{cases}
\end{equation}

Operators $\Ts:\Xc_1\rightarrow\Xc_2$ between two spaces $\Xc_1,\Xc_2$ are usually denoted with block capital letters; for instance, the linear operator associated to the identity matrix $I$ is defined by $\Is:\xb \mapsto I\xb$, but in general operators may be nonlinear. When $\Xc_2\equiv \rea$ sometimes block lowercase letters are used instead, e.g., $\ts:\Xc\rightarrow\rea$.
Given a self-operator $\Ts:\Xc\rightarrow\Xc$, $\fix (\Ts)=\{\xb\in\Xc:\Ts(\xb)=\xb\}$ denotes the set of fixed points and $\zer (\Ts)=\{\xb\in\Xc:\Ts(\xb)=\bzero\}$ denotes the set of zeros.

A \textit{normed space} is a pair $(\Xc,\norm{\cdot})$ where $\Xc$ is a vector space and $\norm{\cdot}$ is a norm on $\Xc$, which induces in the natural way a metric, i.e., a notion of distance: the distance between two vectors $\xb,\yb\in\Xc$ is given by $\norm{\xb-\yb}$.
We will focus on the real vector space $\Xc=\mathbb{R}^n$, thus making any space $(\rea^n,\norm{\cdot})$ a \textit{Banach space} since every finite-dimensional normed vector space is complete (cfr.~\cite[Def. 1.5 and Rem. 2 on page 7]{Berinde07} and~\cite[Theorem 5.33]{Hunter2001}). 
Hilbert spaces are Banach spaces where the inner product is well defined, i.e., $\langle \xb,\xb\rangle =\norm{\xb}^2$ for any $\xb\in\rea^n$.
We will be specifically interested in real vector Banach spaces $(\rea^n,\norm{\cdot}_{\infty,[\etab]^{-1}})$ equipped with diagonally weighted $\ell_{\infty}$ norms defined by positive vectors $\etab\in\rea^n_+$ as follows
$$
\norm{\xb}_{\infty,[\etab]^{-1}} = \norm{[\etab]^{-1}\xb}_{\infty} = \max_{i=1,\cdots,n} \frac{1}{\eta_i} \abs{x_i},
$$
but we will also refer to Hilbert spaces $(\rea^n,\norm{\cdot}_{2,P})$ equipped with weighted $\ell_{2}$ norms defined by positive definite symmetric matrices $P \in \mathbb{S}^n_{\succ 0}$ as follows
$$
\norm{\xb}_{2,P} = \norm{P\xb}_{2} = \sqrt{\xb^\top P^2 \xb}.
$$

We now introduce the basic notions of contractivity and weak contractivity for general Banach spaces.
\begin{defn}
\cite{berinde2021fixed} Given a real Banach space ${(\rea^n,\norm{\cdot})}$, an operator ${\Ts:\rea^n\rightarrow\rea^n}$ is called $\ell\text{-Lipschitz}$ if, for some $\ell \geq0$ and for all $\xb,\yb \in \rea^n$, it satisfies
\begin{equation}\label{eq:lip}
\norm{\Ts(\xb)-\Ts(\yb)} \leq \ell \norm{\xb-\yb}.
% \norm{b(\xb-\yb)+\Ts(\xb)-\Ts(\yb)} \leq \ell (b+1)\norm{\xb-\yb}.
\end{equation}
If we let $\Lip(\Ts)$ be the minimum (or infimum) constant $\ell\geq 0$ which satisfies~\eqref{eq:lip}, i.e.,
$$
\Lip(\Ts) := \sup_{\xb\neq \yb} \frac{\norm{\Ts(\xb)-\Ts(\yb)}}{\norm{\xb-\yb}}
$$
then:
\begin{itemize}
    \item $\Ts$ is \virg{$\ell$-contractive} if $\Lip(\Ts) \in(0,1)$;
    \item $\Ts$ is  \virg{weakly contractive} if $\Lip(\Ts)=1$.
\end{itemize}
\end{defn}

\begin{rem}\label{rem:esup}
    For a Lipschitz operator $\Ts$, the Jacobian matrix $D\Ts(\xb)$ exists for almost every $\xb\in\rea^n$ by Rademacher's theorem. Moreover, it holds that
    $$
    \Lip(\Ts) = \esup_{\xb\in\rea^n} \norm{D\Ts(\xb)},
    $$
    where the essential supremum ignores the points in the set of Lebesgue measure zero where $D\Ts(\xb)$ does not exist.
\end{rem}
Next, we propose a generalization of the weak contractivity property called \textit{enriched weak contractivity}.
\begin{defn}\label{def:seNe}
Given a real Banach space $\Bc={(\rea^n,\norm{\cdot})}$, an operator ${\Ts:\rea^n\rightarrow\rea^n}$ is called $(b,c)\text{-enriched weakly contractive}$ if, for some $b\geq0$, $c\in[0,b+1]$ and for all $\xb,\yb \in \rea^n$, it satisfies
\begin{equation}\label{eq:sene}
\norm{b(\xb-\yb)+\Ts(\xb)-\Ts(\yb)} \leq (b-c+1)\norm{\xb-\yb}.
\end{equation}
%
% If $c=0$, $\Ts$ is simply called $b$-enriched nonexpansive.
%
\end{defn}
We note that for $c\in(b,b+1]$, the coefficient on the right-hand side $(b-c+1)$ is strictly less than one, thus resulting in a weak form of contractivity.
Enriched weak contractivity generalizes the notion of \textit{enriched nonexpansiveness} proposed by Berinde in~\cite{berinde2021fixed}, which is found as a special case when $c=0$; thus, the coefficient on the right-hand side $b+1$ is greater than or equal to one.
Let us collect some useful special cases in the following remark.
\begin{rem}\label{rem:specialcases}
    A $(b,c)$-enriched weakly contractive operator~is:
    \begin{itemize}
        \item Contractive if $0=b<c<1$;
        \item Weakly contractive if $b=c=0$;
        \item Enriched nonexpansive if $b>c=0$~\cite[Eq.~(16)]{berinde2021fixed}.
    \end{itemize}
\end{rem}

\subsection{Relationship with pseudocontractivity}
Another important generalization of weakly contractive operators is that of strictly pseudocontractive operators~\cite{Belgioioso18}. We formally define this property for Hilbert spaces, as its general form for Banach spaces involves more advanced mathematical tools~\cite{deplano25withapp}.
\begin{defn}\label{def:sPC}
\cite[Definition 4]{Belgioioso18} Given a real Hilbert space ${(\rea^n,\norm{\cdot})}$, an operator ${\Ts:\rea^n\rightarrow\rea^n}$ is called $\kappa$-strictly pseudocontractive if, for some $\kappa\in(0,1)$ and for all ${\xb,\yb \in \rea^n}$, it satisfies
\begingroup
\medmuskip=3mu
\thinmuskip=3mu
\thickmuskip=3mu
\begin{equation*}
\norm{\Ts(\xb)-\Ts(\yb)}^2 \leq  \norm{\xb-\yb}^2 + \kappa \norm{\xb-\yb-(\Ts(\xb)-\Ts(\yb))}^2.
\end{equation*}
\endgroup
%
% then it  ($\ell$-\sPC).
\end{defn}

Notably, we show that in the case of Hilbert spaces equipped with an inner product $\langle\cdot,\cdot\rangle$, $\kappa$-strict pseudocontractivity is equivalent to $(b,0)$-enriched weak-contractivity with a specific relation between the parameters $\kappa$ and $b$. On the other hand, for general Banach spaces strict pseudocontractivity and enriched weak contractivity are independent properties. 
\begin{prop}\label{prop:eqKPSENE}
    Let $\Hc=(\rea^n,\norm{\cdot})$ be a real Hilbert space. An operator $T:\rea^n\rightarrow\rea^n$ is $\kappa$-strictly pseudocontractive if and only if it is $(b,0)$-enriched weakly contractive with
    $$
    b = \frac{k}{1-k},\quad \text{or equivalently} \quad k=\frac{b}{b+1}.
    $$
\end{prop}
\begin{proof}
The following transformations hold in both directions, where the first inequality is the definition of enriched weak contractivity and the last is the definition of strict pseudocontractivity:
\begingroup
\medmuskip=0mu
\thinmuskip=0mu
\thickmuskip=0mu
$$
\begin{aligned}
\norm{b(\xb-\yb)+\Ts(\xb)-\Ts(\yb)} & \leq (b-c+1)\norm{\xb-\yb},\\
%%%
\norm{b(\xb-\yb)+\Ts(\xb)-\Ts(\yb)} & \leq (b+1)\norm{\xb-\yb},\\
%%%
% \norm{b(\xb-\yb)+\Ts(\xb)-\Ts(\yb)}^2 & \leq (b+1)^2\norm{\xb-\yb}^2,\\
%%%
\norm{b(\xb-\yb)+\Ts(\xb)-\Ts(\yb)}^2 & \leq (b^2+2b+1)\norm{\xb-\yb}^2,\\
%%%
% \norm{b(\xb-\yb)+\Ts(\xb)-\Ts(\yb)}^2 - b^2\norm{\xb-\yb}^2& \leq (2b+1)\norm{\xb-\yb}^2,\\
%%%
\norm{\Ts(\xb)-\Ts(\yb)}^2
%%%
& \leq (2b+1)\norm{\xb-\yb}^2 \\ & \quad -2b\langle \xb-\yb,\Ts(\xb)-\Ts(\yb)\rangle,\\
%%%
\norm{\Ts(\xb)-\Ts(\yb)}^2 & \leq (2b+1)\norm{\xb-\yb}^2\\
& \quad -2b\langle \xb-\yb,\Ts(\xb)-\Ts(\yb)\rangle\\
& \: + b \norm{\Ts(\xb)-\Ts(\yb)}^2 -b \norm{\Ts(\xb)-\Ts(\yb)}^2 \\
%%%
(b+1)\norm{\Ts(\xb)-\Ts(\yb)}^2 & \leq (b+1)\norm{\xb-\yb}^2 \\
& \quad + b\norm{\xb-\yb-(\Ts(\xb)-\Ts(\yb))}^2,\\
%%%
\norm{\Ts(\xb)-\Ts(\yb)}^2 & \leq \norm{\xb-\yb}^2 \\
& \quad + \frac{b}{b+1}\norm{\xb-\yb-(\Ts(\xb)-\Ts(\yb))}^2.\\
\end{aligned}
$$
We conclude that $\Ts$ is $(b,0)$-enriched weakly contractive if and only if it is $\kappa$-strictly pseudocontractive with $\kappa = b/(b+1)$, thus completing the proof.
\endgroup
\end{proof}

\subsection{Verify enriched weak contractivity w.r.t. $\norm{\cdot}_{\infty,[\etab]^{-1}}$}
The following theorem provides a necessary and sufficient condition for enriched weak contractivity of Lipschitz operators.
\begin{thm}\label{thm:inf-coenex}
Let $\Ts:\rea^n\rightarrow\rea^n$ be a Lipschitz operator.
For $b\geq 0$, $\etab\in\rea^n_+$, the following statements are equivalent:
\begin{itemize}
    \item[(i)] $\Ts$ is $(b,c)$-enriched weakly contractive w.r.t. $\norm{\cdot}_{\infty,[\etab]^{-1}}$;
    \item[(ii)] $\abs{b I + D\Ts(\xb)}\etab \leq (b-c+1)\etab$ for all $\xb\in\rea^n$.
\end{itemize}
Let $b^\star$ be the minimum $b$ such that the above hold, then
$$
0\leq b^\star \leq \max\{0,\normalfont{\text{diagL}(-\Ts)}\},
$$
where
\begin{equation}\label{eq:diagL}
\text{\textnormal{diagL}}(-\Ts):=\esup_{\xb\in\rea^n} \max_{i\in\{1,\ldots,n\}} (-\Ts(\xb))_{ii}.
\end{equation}
\end{thm}
\begin{proof}
    According to the definition of enriched weak contractivity in Definition~\ref{def:seNe}, condition $(i)$ means that the operator $b\Is + \Ts$ is Lipschitz with constant (see also Remark~\ref{rem:esup})
    %
    % $\Lip(\Ts)=1-c/(b+1)$.
    % %
    % Since, according to Remark~\ref{rem:esup}, it holds that
    %
    $$
    \Lip(\Ts) = \esup_{\xb\in\rea^n} \norm{D\Ts(\xb)} = b-c+1.
    $$
    Then, condition $(i)$ is equivalent to
    \begin{equation}\label{eq:suplip}
        \norm{D\Ts(\xb)}_{\infty,[\etab]^{-1}} \leq b-c+1, \text{ for almost every }  \xb\in\rea^n.
    \end{equation}
    Thus, the equivalence $(i)\Leftrightarrow (ii)$ follows from 
    $$
    \begin{aligned}
        \norm{bI + D\Ts(\xb)}_{\infty,[\etab]^{-1}} &\leq b-c+1\\
        \norm{[\etab]^{-1}(bI + D\Ts(\xb))[\etab]}_{\infty} &\leq b-c+1\\
        \abs{[\etab]^{-1}(bI +D\Ts(\xb))[\etab]}\bone &\leq (b-c+1)\bone,\\
        [\etab]^{-1}\abs{bI + D\Ts(\xb)}[\etab]\bone &\leq (b-c+1)\bone,\\
        \abs{bI + D\Ts(\xb)}\etab &\leq (b-c+1)\etab.
    \end{aligned}
    $$
    For the last statement, let $b',c'>0$ be values of $b,c$ such that $(ii)$ holds, and consider the $i$-th row of $(ii)$, namely
    $$
    \abs{b' + (D\Ts(\xb))_{ii}}\eta_i + \sum_{j\neq i}\abs{(D\Ts(\xb))_{ij}}\eta_j \leq (b'-c'+1)\eta_i.
    $$
    It can be noticed that if $b'\geq \text{diagL}(-\Ts)$ the argument of the first absolute value is nonnegative, namely,
    $$
    b'+(D\Ts(\xb))_{ii}\geq b' - \text{diagL}(-\Ts) \geq 0,
    $$
    thus implying that condition $(ii)$ becomes independent from $b'$:
    $$
    (D\Ts(\xb))_{ii}\eta_i + \sum_{j\neq i}\abs{(D\Ts(\xb))_{ij}}\eta_j \leq (1-c)\eta_i.
    $$
    This means that if $(i)$-$(ii)$ hold for any $b'\geq \text{diagL}(-\Ts)$, then they hold also for any other $b\geq \text{diagL}(-\Ts)$, and therefore $b^* \leq \text{diagL}(-\Ts)$. The proof is completed by considering the complementary case $b^\star \leq b'< \text{diagL}(-\Ts)$.
\end{proof}
\begin{cor}\label{cor:inf-conex}
Let $\Ts:\rea^n\rightarrow\rea^n$ be a Lipschitz operator.
For $\etab\in\rea^n_+$, the following statements are equivalent:
\begin{itemize}
    \item[(i)] $\Ts$ is weakly contractive w.r.t. $\norm{\cdot}_{\infty,[\etab]^{-1}}$;
    \item[(ii)] $\abs{D\Ts(\xb)}\etab \leq \etab$ for almost every $\xb\in\rea^n$.
\end{itemize}
\end{cor}

\begin{cor}\label{cor:inf-conex2}
Let $\Ts:\rea^n\rightarrow\rea^n$ be a Lipschitz operator.
For $\etab\in\rea^n_+$, the following statements are equivalent:
\begin{itemize}
    \item[(i)] $\Ts$ is contractive w.r.t. $\norm{\cdot}_{\infty,[\etab]^{-1}}$;
    \item[(ii)] $\abs{D\Ts(\xb)}\etab < \etab$ for almost every $\xb\in\rea^n$.
\end{itemize}
\end{cor}

\subsection{Relationship with monotone operators w.r.t. $\norm{\cdot}_{\infty,[\etab]^{-1}}$}
In this section, we compare the enriched weak contractivity property to that of monotonicity~\cite[Definition 12.1]{rockafellar2009variational} in Banach spaces of the type $(\rea^n,\norm{\cdot}_{\infty,[\etab]^{-1}})$. In the view of~\cite{davydov2024non}, we report the definition of monotonicity in such spaces and the necessary and sufficient condition for Lipschitz operators.
\begin{defn}\label{def:mono}
\cite[Definition 12 and Equation (5)]{davydov2024non} An operator $\Fs:\rea^n\rightarrow\rea^n$ is called $c$-strongly monotone w.r.t. $\norm{\cdot}_{\infty,[\etab]^{-1}}$ if for all $\xb,\yb \in \rea^n$ it holds%there exists a compatible weak pairing $\wp{\cdot,\cdot}$ and $c>0$ such that
%
% \begin{equation}\label{eq:mono}
%     -\wp{-(\Ts(\xb)-\Ts(\yb)), \xb-\yb}\geq 0 %c\norm{x-y}^2,
%     \quad \forall\xb,\yb \in \rea^n.
% \end{equation}
\begingroup
\medmuskip=2mu
\thinmuskip=2mu
\thickmuskip=2mu
\begin{equation}\label{eq:mono}
    \min_{i\in I_{\infty}([\etab]^{-1}\yb)}\frac{(\Fs_i(\xb)-\Fs(\yb))(\xb_i-\yb_i)}{\eta_i^2}\geq c\norm{\xb-\yb}_{\infty,[\etab]^{-1}}.
\end{equation}
\endgroup
where $I_{\infty}(\vb) = \{i\in\{1,\ldots,n\}\mid \abs{\vb_i}=\norm{\vb}_{\infty}\}$.
If~\eqref{eq:mono} holds, then for $c=0$, the operator $\Fs$ is called monotone.
\end{defn}
%
% For (locally) Lipschitz operators, a necessary and sufficient condition si provided in~\cite{davydov2024non}.
%
% \diego{Introdurre weighted norms?}
%
\begin{prop}\label{prop:monoc1}
\cite[Lemma 14]{davydov2024non} A Lipschitz operator ${\Fs:\mathbb{R}^n\rightarrow \mathbb{R}^n}$ is $c$-strongly monotone w.r.t. $\norm{\cdot}_{\infty,[\etab]^{-1}}$ if and only if
\begin{equation}\label{eq:mono_cond}
   \mzr{-D\Fs(\xb)}\etab \leq -c \etab
   \text{ for almost every } \xb \in \rea^n,
\end{equation}
where $\mzr{M}$ is the Metzler majorant of the matrix $M$ as in~\eqref{eq:majorant}.
\end{prop}

Notably, we show that the class of strongly monotone operators is exactly the same class of enriched weakly contractive operators. However, as discussed in more detail later, enriched weak contractivity is a more insightful property due to the extra parameter $b$.
\begin{thm}\label{thm:mono-ene}
    Let $\Ts:\rea^n\rightarrow\rea^n$ be a Lipschitz operator. Consider the following statements:
    \begin{itemize}
        \item[(i)] $\Ts$ is $(b,c)$-enriched weakly contractive w.r.t. $\norm{\cdot}_{\infty,[\etab]^{-1}}$;
        \item[(ii)] $\Fs:=\Is-\Ts$ is $c$-strongly monotone w.r.t. $\norm{\cdot}_{\infty,[\etab]^{-1}}$.
    \end{itemize}
    Then, the following hold:
    \begin{itemize}
        \item[(a)] $(i) \Rightarrow (ii)$ holds for all $b\geq 0 $;
        \item[(b)] $(i) \Leftarrow  (ii)$ holds for $b\geq \textnormal{diagL}(-\Ts)$.
    \end{itemize}
\end{thm}

\begin{proof}
    We start by recalling the necessary and sufficient condition for $(ii)$ enriched weak contractivity in Theorem~\ref{thm:inf-coenex}:
    $$
    \abs{b I + D\Ts(\xb)}\etab \leq (b-c+1)\etab \text{ for almost every } \xb\in\rea^n,
    $$
    and the necessary and sufficient condition for $(ii)$ strong monotonicity of $\Fs$ in Proposition~\ref{prop:monoc1}:
    % By definition, $\Fs$ is monotone if and only if 
    % %
    $$
    \mzr{-D\Fs(\xb)}\etab \leq -c \etab,\text{ for almost every }\xb \in \rea^n.
    $$
    We prove statement $(a)$ by the following steps:
    \begingroup
    $$
    \begin{aligned}
            \mzr{-D\Fs(\xb)}\etab  &= \mzr{D\Ts(\xb)-I}\etab \\
            &= \mzr{D\Ts(\xb)}\etab - \etab \\
            & = \mzr{D\Ts(\xb)}\etab - \etab + b\etab -  b\etab \\
            & = \mzr{bI+D\Ts(\xb)}\etab - (b+1)\etab \\
            & \leq \abs{bI+D\Ts(\xb)}\etab - (b+1)\etab \leq -c\etab.
    \end{aligned}
    $$
    \endgroup
    %
    % where the last inequality holds for some $b>0$ due to Theorem~\ref{thm:inf-coenex} because operator $\Ts$ is assumed to be enriched nonexpansive. 

    Next, we prove statement $(b)$ as follows:
    $$
    \begin{aligned}
    \abs{bI + D\Ts(\xb)}\etab & =  \abs{(b+1)I- D\Fs(\xb)}\etab\\
    & \overset{(i)}{=}  \mzr{(b+1)I- D\Fs(\xb)}\etab\\
    & =  (b+1)\etab + \mzr{- D\Fs(\xb)} \etab\\
    & \leq (b-c+1)\etab
    \end{aligned}
    $$
    \begingroup
\medmuskip=1mu
\thinmuskip=1mu
\thickmuskip=1mu
    where $(i)$ holds for $b\geq \text{diagL}(\Fs)-1 = \text{diagL}(-\Ts)$ as in~\eqref{eq:diagL}.\endgroup
\end{proof}

% \diego{The notion of $\kappa$-strictly pseucontractivity of an operator on a general Banach space is missing in popular books such as those by Berinde~\cite{Berinde07}, Agarwal, Regan and Sahu~\cite{Agarwal09}, Bauschke and Combettes~\cite{Bauschke2017}. Neverthless, pseudocontractivity corresponds to Definition~1.13(b) in~\cite{Berinde07}, but also to~\cite[Eq. (5.57)]{Agarwal09} and to~\cite[Eq. (20.5)]{Bauschke2017}.

% \begin{rem}
% When $p=2$, the space $\Sc_2 = (\rea^n,\norm{\cdot}_2)$ becomes an Hilber space and, in turn, the generalized duality operator only contains the identity operator according to Lemma~\ref{lem:exdualmap}, i.e., $\Ls_r(\Ts(\xb)-\Ts(\yb))=\Ts(\xb)-\Ts(\yb)$ in~\eqref{eq:spc}. In this case, $\kappa$-strict pseudocontractivity reduces to (cfr.~\cite[Eq. (1.3)]{zhou2008convergence})
% %
% \begingroup
% \medmuskip=0mu
% \thinmuskip=0mu
% \thickmuskip=0mu
% $$
% \langle\xb-\yb-(\Ts(\xb)-\Ts(\yb)),\xb-\yb \rangle \geq \frac{1-\kappa}{2} \norm{\xb-\yb-(\Ts(\xb)-\Ts(\yb))}_2^2,
% $$
% \endgroup
% %
% which is the definition of $\alpha$-cocoerciveness~\cite[Definition 4.4]{Bauschke2017} for the map $\Is-\Ts$ where $\alpha = (1-\kappa)/2$. Moreover, cocoercivity and strong monotonicity are dual, i.e., $\Is-\Ts$ is cocoercive if and only if $(\Is-\Ts)^{-1}$ is strongly order-preserving~\cite[Definition 22.1]{Bauschke2017}.
% %
% \end{rem}}

% \subsection{An example}

\subsection{Relationship with monotone systems w.r.t $\norm{\cdot}_{\infty,[\etab]^{-1}}$}

In this section, we also compare enriched weak contractivity with \textit{order-preservation}, which makes discrete-time dynamical system $\xb(k+1)=\Ts(\xb(k))$ monotone~\cite{Angeli03,Deplano23} in the sense of Kamke-Muller~\cite{muller1927fundamentaltheorem,Kamke32}.
%
%We use the former notation in this manuscript to avoid confusion with the monotonicity notion in functional analysis considered in the precious section. 
%
Real vector spaces of finite-dimension can be equipped with the natural order relation $\leq$, yielding ordered vector spaces whose positive cone is the nonnegative orthant ${\mathbb{R}^n_{\geq 0}=\{\xb\in\rea^n\mid \xb\geq 0\}}$.
If between any ${\xb,\yb\in\mathbb{R}^n_{\geq 0}}$ for which there exists an order relation, the operator ${\Ts:\rea^n\rightarrow\rea^n}$ is such that this relation is preserved for their images $\Ts(\xb)$ and $\Ts(\yb)$, then $\Ts$ is said to be \emph{order-preserving}.

\begin{defn}\label{def:op}
\cite[Definition 3]{Deplano20} An operator $\Ts:\mathbb{R}^n\rightarrow \mathbb{R}^n$ is called order-preserving if it holds
\begin{equation}\label{eq:op}
\xb\leq \yb \Rightarrow \Ts(\xb) \leq \Ts(\yb), \qquad \forall\xb,\yb \in \rea^n.
\end{equation}
\end{defn}
\begin{defn}\label{def:opds}\label{def:monosys}
\cite[Definition 1]{Deplano20} A discrete-time dynamical system ${\xb(k+1)=\Ts(\xb(k))}$ is called monotone if the operator $\Ts:\mathbb{R}^n\rightarrow \mathbb{R}^n$ is \virg{order-preserving}.
\end{defn}
\begin{prop}\label{prop:op}
   ~\cite[Theorem 5]{Deplano23} A Lipschitz operator ${\Ts:\rea^n\rightarrow\rea^n}$ is order-preserving if and only if its Jacobian matrix is nonnegative almost everywhere, i.e.,
    $$
    (D\Ts(\xb))_{ij}\geq 0\text{ for almost every } \xb\in\rea^n,\: i,j\in\{1,\ldots,n\}.
    $$
\end{prop}

Monotone dynamical systems are of interest because they enjoy the enriched weak contractivity property w.r.t. a diagonally weighted $\ell_\infty$ if and only if the strict subhomogeneity property holds, which is defined next.
\begin{defn}\label{def:subhomo}
An operator $\Ts:\mathbb{R}^n\rightarrow \mathbb{R}^n$ is called $c\text{-strictly}$ $\etab\text{-subhomogeneous}$ if, for some $c\geq 0$ and for all $\xb\in\rea^n$, it satifies
\begin{equation}\label{eq:subhomo}
    \Ts(\xb+\theta\etab) \leq \Ts(\xb) + \theta(1-c)\etab,\qquad \forall \theta > 0.
\end{equation}
If~\eqref{eq:subhomo} holds with the equality sign, then $\Ts$ is called $\etab$-homogeneous. If~\eqref{eq:subhomo} holds for $c=0$, then $\Ts$ is called $\etab\text{-(sub)homogeneous}$.
\end{defn}
\begin{prop}\label{prop:sh}
A Lipschitz operator ${\Ts:\rea^n\rightarrow\rea^n}$ is $c\text{-strictly}$ $\etab\text{-subhomogeneous}$ if and only if its Jacobian matrix satisfies
$$
D\Ts(\xb)\etab \leq (1-c)\etab,\text{ for almost every } \xb\in\rea^n.
$$
For $c$-strictly $\etab$-homogeneous operators, the necessary and sufficient condition is the above with the strict equality.
\end{prop}

\begin{proof}
    The necessity of the condition is proven via the definition of the directional derivative:
    $$
    \begin{aligned}
    D\Ts(\xb)\etab & = \lim_{\theta\rightarrow 0^+} \frac{\Ts(\xb+\theta\etab) - \Ts(\xb)}{\theta}\\
    & \leq \lim_{\theta\rightarrow 0^+} \frac{\Ts(\xb) + \theta(1-c)\etab - \Ts(\xb)}{\theta} = (1-c)\etab.
    \end{aligned}
    $$
    The sufficiency of the condition is proven by the Newton-Leibnitz formula for vector-valued functions:
    $$
    \begin{aligned}
    f(\xb+\theta\etab) - f(\xb) & = \theta\int_0^1 Df(\xb + s\theta\etab) \etab ds\\
    f(\xb+\theta\etab) - f(\xb) & \leq \theta\int_0^1 (1-c)\etab ds\\
    f(\xb+\theta\etab) - f(\xb) & \leq \theta(1-c)\etab.
    \end{aligned}
    $$
\end{proof}

This allows us to prove the following technical result.
\begin{thm}\label{thm:monoKSPC}
~
Let $\Ts:\rea^n\rightarrow\rea^n$ be a Lipschitz operator such that its Jacobian matrix $D\Ts(\xb)$ is Metzler and consider the following statements:
\begin{itemize}
    \item[(i)] $\Ts$ is $(b,c)$-enriched weakly contractive w.r.t. $\norm{\cdot}_{\infty,[\etab]^{-1}}$;
    \item[(ii)] $\Ts$ is $c$-strictly $\etab\text{-subhomogeneous}$.
\end{itemize}
If $b\geq \textnormal{diagL}(-\Ts)$, then the above conditions are equivalent.
\end{thm}

\begin{proof}
We prove $(i)\Rightarrow (ii)$ by
\begingroup
$$
\begin{aligned}
        D\Ts(\xb)\etab  &= D\Ts(\xb)\etab +b\etab -  b\etab\\
        &= b\etab + D\Ts(\xb)\etab - b\etab \\
        & \overset{(\star)}{=} \abs{b I + D\Ts(\xb)}\etab - b\etab \\
        & \leq (b-c+1)\etab - b\etab \\
        & \leq (1-c)\etab.
\end{aligned}
$$
\endgroup
where $(\star)$ holds by assumption $b\geq \textnormal{diagL}(-\Ts)$.
%
% where the last inequality holds for some $b>0$ due to Theorem~\ref{thm:inf-coenex} because operator $\Ts$ is assumed to be enriched nonexpansive. 

We prove $(i)\Leftarrow (ii)$  by
$$
\begin{aligned}
\abs{bI + D\Ts(\xb)}\etab & \overset{(\star)}{=}  b\etab + D\Ts(\xb)\etab\\
& \leq b\etab + (1-c)\etab\\
& = (b-c+1)\etab
\end{aligned}
$$
where $(\star)$ holds by assumption $b\geq \textnormal{diagL}(-\Ts)$.
% Let us define
% %
% $$
% \Ts_{\frac{1}{b+1}} = \frac{1}{b+1}(b\Is + \Ts)
% $$
% %
% and consider the following statements:
% \begin{itemize}
%     \item[(a)] $\Ts_{\frac{1}{b+1}}$ is (strictly) $\etab\text{-subhomogeneous}$.
%     \item[(b)] $\Ts_{\frac{1}{b+1}}$ is (contractive) weakly contractive w.r.t. $\norm{\cdot}_{\infty,[\etab]^{-1}}$.
% \end{itemize}

% We start by noticing that strict subhomogeneity of $\Ts$ is equivalent to that of $\Ts_{\frac{1}{b+1}}$, i.e., $(ii)\Leftrightarrow(a)$.
% %
% By construction, the Jacobian matrix $D\Ts_{\frac{1}{b+1}}$ is Metzler because $D\Ts$ is assumed to be Metzler.
% %
% Moreover, the diagonal elements of $D\Ts_{\frac{1}{b+1}}$ are lower-bounded by ${(b-\normalfont{\text{diagL}(-\Ts)}})/(b+1)$ -- 
% due to the Lipschitzness of $\Ts$-- and are nonnegative for ${b\geq \normalfont{\text{diagL}(-\Ts)}}$.
% %
% Under this condition, by Proposition~\ref{prop:op}, $\Ts_{\frac{1}{b+1}}$ is order-preserving.
% %
% Thus, Lemma~\ref{lem:supnormne} implies $(a)\Leftrightarrow (b)$.
% %
% Since Lemma~\ref{lem:lipcond} ensures $(i)\Leftrightarrow(b)$, the proof is concluded by $(i)\Leftrightarrow(b)\Leftrightarrow(a)\Leftrightarrow(ii)$.
\end{proof}

\begin{cor}\label{cor:monoKSPC}
~
Consider a discrete-time dynamical system ${\xb(k+1)=\Ts(\xb(k))}$ and assume it is monotone.
%
%is Metzler and its diagonal elements are lower bounded by a nonpositive constant $L\leq 0$. 
%
Then, the following statements are equivalent:
\begin{itemize}
    \item[(i)] $\Ts$ is $(b,c)$-enriched weakly contractive w.r.t. $\norm{\cdot}_{\infty,[\etab]^{-1}}$;
    \item[(ii)] $\Ts$ is $c$-strictly $\etab\text{-subhomogeneous}$.
\end{itemize}
\end{cor}

\section{The Krasnoselskij iteration\\ in a Banach space}\label{sec:krasno}

\subsection{Main convergence result}

Consider the Krasnoselskij iteration,
\begin{equation}\label{eq:mainiter}
    \xb(k+1) =\Ts_\theta(\xb(k)) = (1-\theta)\xb(k)+\theta \Ts(\xb(k)),
\end{equation}
with step size $\theta \in(0,1)$. The next lemma provides an upper bound on the maximum value of $\theta$ guaranteeing the convergence of the iteration, provided that the operator $\Ts$ is enriched weakly contractive as in Definition~\ref{def:seNe}.

\begin{lem}\label{lem:lipcond}
Consider a Banach space $(\rea^n,\norm{\cdot})$ and an operator $\Ts:\rea^n\rightarrow \rea^n$. The following statements are equivalent:
\begin{itemize}
 \item $\Ts$ is $(b,c)$-enriched weakly contractive w.r.t. $\norm{\cdot}$;
 \item $\Ts_\theta$ in~\eqref{eq:mainiter} is Lipschitz for $\theta=1/(b+1)$ w.r.t. $\norm{\cdot}$ and with constant $\ell=1-\frac{c}{b+1}$.
\end{itemize}
\end{lem}
\begin{proof}
The following transformations hold in both directions, where the first inequality is the definition of enriched weak contractivity and the last is the definition of Lipschitzianity for $\Ts_{\theta}$ with $\theta = 1/(b+1)$:
\begingroup
\medmuskip=2mu
\thinmuskip=2mu
\thickmuskip=2mu
\allowdisplaybreaks
\begin{align*}
\norm{b(\xb-\yb)+\Ts(\xb)-\Ts(\yb)} & \leq (b-c+1)\norm{\xb-\yb},\\
%%%
\frac{1}{b+1}\norm{b(\xb-\yb)+\Ts(\xb)-\Ts(\yb)} & \leq \frac{b-c+1}{b+1}\norm{\xb-\yb},\\
%%%
\norm{\frac{b(\xb-\yb)+\Ts(\xb)-\Ts(\yb)}{b+1}} & \leq \left(1-\frac{c}{b+1}\right)\norm{\xb-\yb},\\
%%%
\norm{\frac{b\xb+\Ts(\xb)}{b+1} - \frac{b\yb+\Ts(\yb)}{b+1}} & \leq \left(1-\frac{c}{b+1}\right)\norm{\xb-\yb},\\
%%%
\norm{\Ts_{\frac{1}{b+1}}(\xb) - \Ts_{\frac{1}{b+1}}(\yb)} & \leq \left(1-\frac{c}{b+1}\right)\norm{\xb-\yb}.
\end{align*}
\endgroup
% We conclude that $\Ts$ is $c$-strongly $b$-enriched nonexpansive if and only if $\Ts_{\theta}$ with $\theta=1/(b+1)$ is Lipschitz with constant $1-c/(b+1)$, thus completing the proof.
\end{proof}

%

%

% \begin{proof}
% %
% Condition $(i)$ means that $\Ts$ is Lipschitz with constant $\Lip(\Ts)=1$. Since, according to Remark~\ref{rem:esup}, it holds that
% %
% $$
% \Lip(\Ts) = \esup_{\xb\in\rea^n} \norm{D\Ts(\xb)},
% $$
% then condition $(i)$ is equivalent to
% %
% \begin{equation}\label{eq:suplip}
%     \norm{D\Ts(\xb)}_{\infty,[\etab]^{-1}} \leq 1, \text{ for almost every }  \xb\in\rea^n.
% \end{equation}
% Thus, $(i)\Leftrightarrow(ii)$ follows from   %
% $$
% \begin{aligned}
%     \norm{D\Ts(\xb)}_{\infty,[\etab]^{-1}} &\leq 1\\
%     \norm{[\etab]^{-1}D\Ts(\xb)[\etab]}_{\infty} &\leq 1\\
%     \abs{[\etab]^{-1}D\Ts(\xb)[\etab]}\bone &\leq \bone,\\
%     [\etab]^{-1}\abs{D\Ts(\xb)}[\etab]\bone &\leq \bone,\\
%     \abs{D\Ts(\xb)}\etab &\leq \etab,\\
% \end{aligned}
% $$
% \end{proof}
%

We now state and prove the main result of the paper.
\begin{thm}\label{thm:nonlin-bene}
    Consider an operator $\Ts$ that is $(b,c)$-enriched weakly contractive w.r.t. $\norm{\cdot}_{\infty,[\etab]^{-1}}$
    for some $\etab \in \mathbb{R}^n_{+}$, $b> 0$, $c\geq 0$. Then the following holds:
    \begin{itemize}
        \item[(a)] If $c=0$ and ${\fix(\Ts)\neq \emptyset}$, the iteration in~\eqref{eq:mainiter} converges for ${\theta\in\left(0,\frac{1}{b+1}\right)}$ to some fixed point $\bar{\xb}\in \fix (\Ts)$;
        %%%
        \item[(b)] If $c>0$, the iteration in~\eqref{eq:mainiter} converges to the unique fixed point $\fix(\Ts)=\{x^\star\}$ for ${\theta\in\left(0,\frac{1}{b+1}\right]}$. Moreover, it holds that
        $$
        \norm{\xb(k+1)-\xb^\star}\leq \left(1-\theta c\right)\norm{\xb(k)-\xb^\star}, \quad \forall k\in\nat.
        $$
        %
        % where $(1-\theta c)$ is an upper bound to the convergence rate.
    \end{itemize}
    %
    % Moreover, It holds that $(a)\Rightarrow(b)$ and $(a)\not\Leftarrow(b)$, even though $\Ts$ is a linear operator.
\end{thm}
\begin{proof}
    The iteration can be rewritten as follows
    $$
    \begin{aligned}
        \Ts_\theta &= (1-\theta)\Is+\theta \Ts = (1-\theta)\Is+\theta \left((b+1)\Ts_{\frac{1}{b+1}}-b\Is\right)\\
          & = (1-\theta(b+1))\Is+\theta (b+1)\Ts_{\frac{1}{b+1}}\\
          & = (1-\alpha)\Is + \alpha\Ts_{\frac{1}{b+1}},\quad \text{with} \quad \alpha =  \theta(b+1).
    \end{aligned}
    $$
    Due to Lemma~\ref{lem:lipcond}, the operator $\Ts_{\frac{1}{b+1}}$ ruling the iteration in~\eqref{eq:mainiter} is either weakly contractive (if $c=0$) or contractive (if $c=0$) with contraction factor $1-c/(b+1)$ or .
    If $c=0$, by the known result of Ishikawa~\cite{ishikawa76} which can also be found in~\cite[Lemma 11]{davydov2024non}, the iteration converges to a fixed point (if any exist) for any $\alpha\in(0,1)$, i.e., $\theta\in(0,1/(b+1))$, thus proving $(a)$.
    If $c>0$, by the known Banach fixed point theorem, the iteration converges to the unique fixed point and the convergence rate is given by 
    \begingroup
    \medmuskip=1mu
    \thinmuskip=1mu
    \thickmuskip=1mu
    \allowdisplaybreaks
    \begin{align*}
    & \norm{D\Ts_{\theta}(\xb)}_{\infty,[\etab]^{-1}} = \norm{(1-\theta) I + \theta D\Ts(\xb)}_{\infty,[\etab]^{-1}}\\
    %%%
    & = \max_{i}\left\{\abs{1-\theta + \theta (D\Ts(\xb))_{ii}} + \theta \sum_{j\neq i} \abs{(D\Ts(\xb))_{ij}}\frac{\eta_j}{\eta_i}\right\}\\
    %%%
    %%%
    & = \max_{i}\left\{\abs{1-\theta+\theta(D\Ts(\xb))_{ii}+\alpha-\alpha} + \theta \sum_{j\neq i} \abs{(D\Ts(\xb))_{ij}}\frac{\eta_j}{\eta_i}\right\}\\
    %%%
    & \overset{(i)}{\leq} 1-\alpha + \theta\max_{i}\left\{\abs{((D\Ts(\xb))_{ii}+b)} + \sum_{j\neq i} \abs{(D\Ts(\xb))_{ij}}\frac{\eta_j}{\eta_i}\right\}\\
    %%%
    & \overset{(ii)}{=} 1-\alpha + \theta(b+1)\norm{D\Ts_{\frac{1}{b+1}}}_{\infty,[\etab]^{-1}}\\
    %%%
    & \overset{(iii)}{\leq} 1-\alpha + \theta(b+1)  \left(1-\frac{c}{b+1}\right)\\
    %%%
    & \overset{}{=} 1-\alpha + \theta(b+1)  -\theta c\\
    %%%
    & \overset{(iv)}{=} 1-\frac{\alpha c}{b+1} = 1-\theta c.
    \end{align*}
    \endgroup
    We note that $(i)$ holds because since $1-\alpha$ with $\alpha\in(0,1)$ is always positive, it can be moved outside the absolute value by the triangle inequality; also, $(i)$ uses the identity ${\alpha-\theta=\theta b}$, thus allowing to take $\theta$ as a common factor; $(ii)$ holds by definition of the diagonally weighted $\ell_{\infty}$ norm and by the definition of ${\Ts_{\frac{1}{b+1}} = (b\cdot \Is +\Ts)/(b+1)}$; $(iii)$ holds by Lemma~\ref{lem:lipcond} which ensures that the Lipschitz constant of $\Ts_{\frac{1}{b+1}}$ is $1-c/(b+1)$; $(iv)$ holds by the identity ${\alpha= \theta(b+1)}$. %This completes the proof.
    %
    % To prove $(a)\not\Leftarrow (b)$ we provide a counter example. We consider a linear operator $\Ts:\xb\mapsto T\xb$ where
    % %
    % $$
    % T = \begin{bmatrix}
    %     1 & b+1\\
    %     0 & a(b+1)-b
    % \end{bmatrix},\quad a\in(0,1),b>0.
    % $$
    % %
    % Due to proposition~\ref{thm:spcIFFne}, $\Ts$ is $b$-enriched nonexpansive if and only if $ T_{\frac{1}{b+1}}$, defined next, is nonexpansive w.r.t. a weighted diagonal norm, which is not possible by Example~\ref{exmp:counter},
    % %
    % $$
    % T_{\frac{1}{b+1}} = \frac{1}{b+1}\left(\begin{bmatrix}
    %     b & 0\\ 0 & b
    % \end{bmatrix} + \begin{bmatrix}
    %     1 & b+1\\
    %     0 & a(b+1)-b
    % \end{bmatrix}\right) = \begin{bmatrix}
    %     1 & 1\\
    %     0 & a
    % \end{bmatrix}.
    % $$
    % Since the iteration converges for any $\theta\in(0,1/(b+1)$ but $\Ts$ is not $b$-enriched nonexpansive w.r.t. $\norm{\cdot}_{\infty,[\etab]^{-1}}$ for any $b>0$ and any $\etab\in\rea^n_+$, then $(a)\not\Leftarrow(b)$, completing the proof.
    %
\end{proof}

%
% What Proposition~\ref{prop:belgio} is saying is that the norm $\norm{\cdot}_{2,P}$ is logaritmically optimal, i.e., the spectral abscissa of $A$ corresponds with the logarithmic norm of $A$.

% We now show that \virg{\emph{strict pseudocontractivity}} and \virg{\emph{enriched nonexpansiveness}} are equivalent properties in Hilbert spaces\footnote{This result is a special case of~\cite[Theorem 8]{berinde2021fixed}, but we provide a novel simpler proof to make this manuscript self-contained.}, while this is not the case for general Banach spaces.
% %

% \diego{Add discussion and result for the optimal convergence rate}

\subsection{The special case of linear operators: A comparison with strict pseudocontractive operators}
We start our discussion by recalling a recent result by Belgioioso \textit{et al.}~\cite{Belgioioso18} in the case of linear operators and Hilbert spaces $(\rea^n,\norm{\cdot}_{2,P})$.
Under these assumptions, the iteration converges if and only if the map $\Ts$ is $k$-strictly pseudocontractive which, in Hilbert spaces, is equivalent to being $(b,0)$-enriched weakly contractive, by Proposition~\ref{prop:eqKPSENE}.
Let us formally report this result in view of Proposition~\ref{prop:eqKPSENE}.
\begin{prop}\label{prop:belgio}
   ~\cite[Theorem 1]{Belgioioso18} Consider the iteration
    $$
    \xb(k+1) = (1-\theta)\xb(k)+\theta A\xb(k),
    $$
    For $b\geq 0$, the following statements are equivalent:
    \begin{itemize}
        \item $\exists P \in \mathbb{S}^n_{\succ 0}:$ $\As:\xb\mapsto A\xb$ is $(b,0)$-enriched weakly contractive w.r.t.~$\norm{\cdot}_{2,P}$;
        \item $\xb(k)$ converges to $\fix(\As)$ for ${\theta\in\big(0,\frac{1}{b+1}\big)}$.
    \end{itemize}

\end{prop}

On the contrary, we show that enriched weak contractivity w.r.t. a diagonally weighted $\ell_{\infty}$ norm is only sufficient for the convergence of the iteration, but not necessary. We do so by providing a counter-example for a weakly contracting matrix, i.e., a $(0,0)$-enriched weakly contracting matrix.

Theorem~\ref{thm:inf-coenex} and Corollary~\ref{cor:inf-conex} provide a method to verify (enriched) weak contractivity of an operator w.r.t. a diagonal weighted $\ell_{\infty}$.
In the case of linear operators ${\As:\xb\mapsto A\xb}$, weak contractivity w.r.t. a diagonally weighted $\ell_{\infty}$ norm reduces to a linear feasibility program (LP) of the kind ${\abs{A}\etab \leq \etab}$; 
on the other hand, weak contractivity w.r.t. a weighted $\ell_2$ is a semidefinite feasibility program (SDP) of the kind ${A^\top P A \preccurlyeq P}$~\cite[Lemma 3]{Belgioioso18}. 
An important consideration due to~\cite[Lemma 3]{Belgioioso18} is that the SDP feasibility is a necessary and sufficient condition for (at most marginal) stability of the system $\xb(k+1)=A\xb(k)$, i.e., all eigenvalues of $A$ are within the unit disk and those on the boundary are semi-simple. 
On the contrary, the next example shows that the LP feasibility is not necessary, while it is sufficient as proved in the following Proposition \ref{prop:suff}.

\begin{exmp}\label{exmp:counter}
Consider the matrix:
$$
A = \begin{bmatrix}
    1 & 1\\ 0 & a
\end{bmatrix}, \quad \text{with} \quad a\in(0,1).
$$
The eigenvalues $\{a,1\}$ of $A$ are semi-simple and within the unit disk, thus the iteration $\xb(k+1)=A\xb(k)$ converges. However, the operator $\As:\xb\rightarrow A\xb$ is not weakly contractive w.r.t. any diagonally weighted $\ell_{\infty}$ norm because the LP program $\abs{A}\etab \leq \etab$ is not feasible for $\etab\in\rea^n_+$:
$$
\begin{cases}
A\etab \leq \etab\\
\etab > \bzero
\end{cases} \Leftrightarrow \begin{cases}
    \eta_1 + \eta_2 \leq \eta_1\\
    a\eta_2 \leq \eta_2\\
    \eta_1 > 0\\
    \eta_2 > 0
\end{cases} \Leftrightarrow \begin{cases}
    \eta_2 \leq 0\\
    a \leq 1\\
    \eta_1 > 0\\
    \eta_2 > 0
\end{cases}.
$$
% \diego{Magari si può aggiungere la matrice $P$ per la quale è nonexpansive con la norma $2$.}
\end{exmp}

The above example shows that weak contractivity w.r.t. a diagonally weighted $\ell_{\infty}$ norm is not necessary for (at least marginal) stability of a linear time-invariant system. In the next proposition, we show that it is sufficient.

\begin{prop}\label{prop:suff}
Let ${\As:\xb\mapsto A\xb}$ be a linear operator. Weak contractivity of $\As$ w.r.t. $\norm{\cdot}_{\infty,[\etab]^{-1}}$ implies that
\begin{itemize}
    \item[($\star$)] All eigenvalues of $A$ are in the unit disk and each eigenvalue on the boundary of the unit disk is semi-simple.
\end{itemize} 
Vice versa, the above condition does not imply weak contractivity of $\As$ w.r.t. $\norm{\cdot}_{\infty,[\etab]^{-1}}$.
\end{prop}
    
\begin{proof}
We first note that weak contractivity is equivalent to $(0,0)$-enriched weak contractivity, which holds if and only if (Corollary~\ref{cor:inf-conex2})
\begin{equation}\label{eq:rs}
\abs{A}\etab\leq \etab
\end{equation}
When $\Ts=\As$ is a linear operator ${\As:\xb\rightarrow A\xb}$, then $D\Ts(\xb) = D\As(\xb) = A$ for all ${\xb\in\rea^n}$.
Consider the similarity transformation given by $N=[\etab]$, i.e., ${A_N = N^{-1}\abs{A}N}$. By construction, matrices $\abs{A}$ and $A_N$ have the same eigenvalues and are both nonnegative.
Exploiting~\eqref{eq:rs}, it can be show that $A_N$ is row-substochastic,
$$
A_N\bone = N^{-1}\abs{A}N\bone  = N^{-1}\abs{A}\etab \leq N^{-1}\etab = \bone,
$$
and, in turn, the eigenvalues of $A$ are in the unit disk:
$$
\rho(A) \leq \rho(\abs{A}) = \rho(A_N) \leq \max(A_N\bone) \leq \max(\bone) \leq 1,
$$
For the sake of contradiction, assume there is an eigenvalue with magnitude $1$ that is not semi-simple, i.e., its algebraic multiplicity is strictly greater than its geometric multiplicity. This implies that the system ${\xb(k)=A\xb(k-1)}$ is unstable, i.e., at least one component of $\xb(k)$ grows unbounded as $k\rightarrow\infty$.
w%
This contradicts the assumption that the linear operator $\As$ is weakly contractive because it implies that $\xb(k)$ must remain bounded. Indeed, weak contractivity of $\As$ is equivalent to ${\norm{A}_{\infty,[\etab]^{-1}} \leq 1}$ as in~\eqref{eq:suplip} and therefore
$$
\begin{aligned}
    \lim_{k\rightarrow\infty} \norm{\xb(k)}_{\infty,[\etab]^{-1}} &\leq \lim_{k\rightarrow\infty} \norm{A^k\xb(0)}_{\infty,[\etab]^{-1}} \\
    & = \lim_{k\rightarrow\infty} \norm{A}^k_{\infty,[\etab]^{-1}} \norm{\xb(0)}_{\infty,[\etab]^{-1}}\\
    & = \norm{\xb(0)}_{\infty,[\etab]^{-1}}. 
\end{aligned}
$$
This proves that weak contractivity of $\As$ implies ($\star$). The contrary does not hold by the counterexample discussed in Example~\ref{exmp:counter}, thus completing the proof.
\end{proof}

\subsection{A comparison with monotone operators and monotone dynamical systems}\label{sec:monotone}

By Theorem~\ref{thm:mono-ene}, we know that $\Fs$ is $c$-strongly monotone, then $\Ts = \Is-\Fs$ is surely $(b,c)$-enriched weakly contractive for any $b\geq\text{diagL}(-\Ts)$.
Thus, the value $\text{diagL}(-\Ts)$ works as an upper bound to minimum value $b^\star$ for which the $\Ts$ is enriched weakly contractive. This is coherent with Theorem~\ref{thm:inf-coenex}. 
This means that the information provided by the monotonicity property is somehow less than that provided by enriched weak contractivity, in the sense that monotonicity of $\Fs$ allows to find an upper bound on the constant of enriched weak contractivity of $\Ts$, but lower values may be still valid.
We provide an example illustrating this case, which is the most common.
%
% \diego{XXX add reference to the simulation section XXX}

\begin{exmp}\label{exmp:largerss}
Consider the matrix
\begin{equation}\label{eq:largerss}
A = \frac{1}{2}\begin{bmatrix*}[r]
    -3&     0&    1&     -3\\
    3&    -15&    -12&     -1\\
    2&     -1&    -5&     -5\\
    -2&     0&     -1&    -6
\end{bmatrix*}.
\end{equation}
Consider the Krasnoselskij iteration:
$$
\xb(k+1) = (1-\theta)\xb(k) + \theta A\xb(k).
$$
According to Theorem~\ref{thm:inf-coenex}, the operator $\As:\xb\mapsto A\xb$ is $(b,0)$-enriched weakly contractive w.r.t. $\norm{\cdot}_{\infty,[\etab]^{-1}}$ if and only ${\abs{bI+A}\etab \leq (b+1)\etab}$. 
It can be verified that this holds with the choices of ${\etab = [0.09,1,0.22,0.07]^\top}$ and $b = 4$, which satisfies the upper bound $\textnormal{diagL(-A)} = 15/2$ given by Theorem~\ref{thm:inf-coenex}.
Therefore, by Theorem~\ref{thm:nonlin-bene} the iteration converges for ${\theta\in(0,\frac{1}{b+1})}$, i.e., ${\theta\in(0,0.2)}$.

The Krasnoselskij iteration can be equivalently re-written as the forward step method on the matrix $F=I-A$:
$$
\begin{aligned}
    \xb(k+1) &= (1-\theta)\xb(k) + \theta A\xb(k)\\
    & = \xb(k) - \theta(I- A)\xb(k) = \xb(k) - \theta F\xb(k).
\end{aligned}
$$
where
$$
F= \frac{1}{2}\begin{bmatrix*}[r]
    5&     0&   -1&    3\\
   -3&    17&    12&    1\\
   -2&     1&    7&    5\\
    2&     0&    1&    8
\end{bmatrix*}.
% F= \begin{bmatrix}
%     +5.5&         0&   -0.5&    +1.5\\
%    -1.5&    +8.50&    +6.0&    +0.5\\
%    -1.0&    +0.50&    +3.5&    +2.5\\
%     1.0&         0&    +0.5&    +4.0
% \end{bmatrix}.
$$
According to~\cite[Theorem 26(iii)]{davydov2024non}, if the operator ${\Fs:\xb\rightarrow F\xb}$ is monotone as in Definition~\ref{def:mono} (with $c=0$), the iteration converges for $\theta\in\big(0,\frac{1}{\textnormal{diagL}(\Fs)}\big)$ where $\textnormal{diagL}(\Fs)=1+\textnormal{diagL}(-A)$.
We note that this range of feasible values of $\theta$ corresponds to the one obtained from the upper bound given in Theorem~\ref{thm:inf-coenex}.
The operator $\Fs$ is monotone if and only if there exists ${\etab=[\eta_1,\cdots,\eta_n]\in\rea^n_{+}}$ such that (cfr. Proposition~\ref{prop:monoc1})
$$
-F_{ii}\eta_i +\sum_{j\neq i} \abs{F_{ij}}\eta_j \leq 0.
$$ 
The above holds for the same ${\etab=[0.09,1,0.22,0.07]^\top}$ and, in turn, the iteration converges for $\theta\in\big(0,0.117\big)$.

To make the results easily verifiable, we did not consider the coefficient of enriched weak contractivity and that of strong monotonicity. A more precise characterization of $\As$ is that it is $(b,c)$-enriched weak contractivity for $b\approx 4.006$ and ${c\approx 0.0227}$, while a more precise characterization of $\Fs$ is that it is $c$-strongly monotone with the same $c\approx 0.0227$.

% \diego{It could be interesting verifying the maximum value of $\theta$ when the norm $\norm{\cdot}_{2,P}$ is considered instead.}
\end{exmp}

% The above example provides an example of an operator  shoes that We conclude that the iteration can be proved to converge for a sufficiently small value of $\theta$ both showing that $A$ is enriched nonexpansive or showing that $(I-A)$ is monotone.
% %
% The above example shows that exploiting the enriched nonexpansiveness properties allows us to select larger values of $\theta$ (in the example, up to $0.2$ instead of $0.117$) while still guaranteeing convergence of the iteration.

% \subsection{Easy-to-verify bound on the step size for monotone dynamical systems}\label{sec:op}
%

We conclude the section with an easy-to-verify bound on the step size for monotone dynamical systems by combining the results of Theorem~\ref{thm:monoKSPC} and Theorem~\ref{thm:nonlin-bene}.
\begin{cor}\label{cor:nonlin-bene-mono}
    Consider a monotone dynamical system with dynamics
    $\xb(k+1) = (1-\theta)\xb(k)+\theta \Ts(\xb(k))$ where $\Ts$ is a Lipschitz operator whose Jacobian matrix is Metzler and let ${L=\normalfont{\text{diagL}(-\Ts)}}$.
    Then, the following hold:
    \begin{itemize}
        \item If $\Ts$ is $\etab$-subhomogeneous and ${\fix(\Ts)\neq \emptyset}$, then the iteration in~\eqref{eq:mainiter} converges to some fixed point $\bar{\xb}\in \fix (\Ts)$ for ${\theta\in\big(0,\frac{1}{1+L}\big)}$.
        %%%
        \item If $\Ts$ is strictly subhomogeneous, the iteration in~\eqref{eq:mainiter} converges to the unique fixed point $\fix(\Ts)=\{x^\star\}$ for ${\theta\in\big(0,\frac{1}{1+L}\big]}$ with convergence rate optimized at ${\theta=1/(1+L)}$.
        %.
    \end{itemize}
    %
    % Moreover, It holds that $(a)\Rightarrow(b)$ and $(a)\not\Leftarrow(b)$, even though $\Ts$ is a linear operator.
\end{cor}
%
% The above result can also be derived as special case of already his result is a special case of~\cite[Theorem 26(ii)-(iii)]{davydov2024non}, because by Theorems~\ref{thm:mono-ene}-\ref{thm:monoKSPC} it holds that any (strict) subhomogeneous operator $\Ts$ with a Metzler Jacobian is such that the operator $\Fs=\Is-\Ts$ is (strongly) monotone, and the results in~\cite[Theorem 26(ii)-(iii)]{davydov2024non} are specialized with $\normalfont{\text{diagL}(\Fs)}=1+\normalfont{\text{diagL}(-\Ts)}=1+L$.

% Moreover, if $\theta\in(0,1-\kappa)$, then $\Ts$ is $\theta$-averaged and the solution of the system converge to some $$\bar{x}\in \fix (\Ts_\theta)=\fix (\Ts).$$

% By Theorem~\ref{thm:spcIFFne}, the $Ts_\theta$ is averaged for $\theta\in(0,\theta^*)$, but since $\theta^* \geq 1-k$ then we can say that $\mathcal{T}_\theta$ is averaged for $\theta\in(0,1-\kappa)$, and the convergence to a fixed point follows from~\cite[Theorem 1]{ishikawa76}

% {\color{red} \begin{rem}
% Actually, the convergence result of Theorem~\ref{th:convergence-kras} can also be obtained by showing that $\Ts_{\theta}$ is type-K order preserving and subhomogeneous, i.e., K-subtopical, and thus exploiting one of the results of the paper TAC under review.
% \end{rem}}

\section{Application to Zero-Finding Algorithms\\ for Monotone Operators}\label{sec:app1}
The problem of finding a zero of an operator (e.g., the zero of the Laplacian operator in consensus problems \cite{bondy2008graph,Deplano20,bronski2014spectral,devries2018kernel,qu2025controllability}) is usually translated into the problem of finding a fixed point of a suitably defined operator. According to Theorem~\ref{thm:mono-ene}, the problem of finding a zero of a (strongly) monotone operator $\Fs$ is equivalent to the problem of finding a fixed point of the enriched weakly contractive operator $\Ts=\Is-\Fs$, namely:
$$
\Fs(\xb) = \bzero \quad \Leftrightarrow \quad \Ts(\xb) = \xb.
$$
The standard \virg{forward step method} to find a zero of the monotone operator $\Fs$ consists in iterating the \virg{forward operator}
\begin{equation}\label{eq:fsm}
    \Ss_{\theta\Fs} := \Is-\theta \Fs,
\end{equation}
\begin{thm}\label{thm:fsm}
    Let $\Fs:\rea^n\rightarrow\rea^n$ be a Lipschitz operator. If $\zer(\Fs)\neq \emptyset $ and $\Fs$ is strongly monotone w.r.t. $\norm{\cdot}_{\infty,[\etab]^{-1}}$, then the iteration ruled by $\xb(k+1)=\Ss_{\theta\Fs}(\xb(k))$ with $\Ss_{\theta\Fs}$ as in~\eqref{eq:fsm} converges to an element of $\zer(\Fs)$ for every 
    $$
    \theta \in\left(0,\frac{1}{b^\star+1}\right),
    $$
    where $b^\star\leq \text{\normalfont{diagL}}(\Fs)-1$ is given by
    $$
    b^\star = \min\{b\geq 0 \abs{(b+1) I - D\Fs(\xb)}\etab \leq (b+1)\etab,\etab>0\}.
    $$
\end{thm}
\begin{proof}
   The forward step method corresponds to the Krasnoselskij iteration with operator $\Ts:=\Is-\Fs$:
    $$
    \Ss_{\theta\Fs} = \Is-\theta \Fs= (1-\theta)\Is+\theta (\Is-\Fs) = (1-\theta)\Is+\theta \Ts. 
    $$
    By Theorem~\ref{thm:mono-ene}, $\Fs$ is $c$-strongly monotone if and only if $\Ts$ is $(b,c)$-enriched weakly contractive for some $b\geq 0$ which, in turn, is equivalent to $\abs{b I + D\Ts(\xb)}\etab \leq (b-c+1)\etab$ for almost every $\xb\in\rea^n$ by Theorem~\ref{thm:inf-coenex}. The smallest admissible value $b^\star$ of $b$ is obtained by selecting $c=0$. Theorem~\ref{thm:nonlin-bene} ensures that the iteration converges to an element of $\fix(\Ts)$ for any $0<\theta < 1/(b+1)\leq 1/(b^\star+1)$. The proof is completed by exploiting Theorem~\ref{thm:inf-coenex} which yields $b^\star\leq  \text{diagL}(-\Ts) = \text{diagL}(\Fs-\Is) = \text{diagL}(\Fs) -1$.
\end{proof}

We now provide some numerical simulations to illustrate that Theorem~\ref{thm:fsm} represents a generalization of~\cite[Theorem 26]{davydov2024non} since it allows for larger step sizes and, in turn, improved convergence rate.

\subsection{Affine operators}\label{sec:aff}
We first consider the simple case where $\Ts:=\As$ is an affine operator ${\As:\xb\mapsto A\xb+\bb}$ and fixed-point problems of the form
\begin{equation*}
\xb(k+1) = \Ts(\xb(k)) := A\xb(k) + \bb.
\end{equation*}
Consider now $A$ and $\bb$ as follows:
$$
A = \left[\begin{array}{rrrr}
    -1.07& -0.17& -0.53& -0.33\\
    0.07&  0.42& -0.07&  0.15\\
    -0.13& -0.10&  -0.06& -0.30\\
    0.04&  0.05& -0.21&  0.40 
\end{array}\right],\quad \bb=\left[\begin{array}{r}
    1\\1\\1\\1
\end{array}\right].
$$
In this case, there is a unique fixed point given by
$$
\xb^\star\approx\begin{bmatrix}
    0.04& -2.14& -0.25& -1.76
\end{bmatrix}^\top.
$$
It is possible to exploit either enriched weak contractivity or strong monotonicity to determine the range of possible values for the Krasnoselskij iteration w.r.t. $\norm{\cdot}_{\infty}$:
\begin{itemize}
    \item \textit{Enriched weak contractivity}: Theorem~\ref{thm:inf-coenex} ensures convergence of the Krasnoselskij iteration for ${0\leq\theta \leq \frac{1}{b^\star+1}\leq 0.645}$, because the lowest value of $b$ such that $\As$ is $(b,0)$-enriched weakly contractive is $b\geq b^\star= 0.55$;
    \item \textit{Monotonicity}: Theorem 26(iii) in \cite{davydov2024non} ensures convergence of the Krasnoselskij iteration for ${0\leq\theta \leq \frac{1}{\text{diagL}(\Is-\As)}= 0.48}$, because $\Is-\As$ is monotone.
\end{itemize}
The above derivations show that exploiting enriched weak contractivity -- instead of monotonicity -- allows to enlarge the range of admissible step sizes ensuring the convergence of the Krasnoselskij iteration, which is coherent with the relationship $b^\star\leq\max\{0,\text{diagL}(-\As)\} =1.07$ provided in Theorem~\ref{thm:inf-coenex}.

We now discuss how to chose $\theta$ to guarantee the best (smallest) convergence rate. 
According to Theorem~\ref{thm:nonlin-bene}, the convergence rate depends on both parameters $b$ and $c$ determining $(b,c)$-enriched weakly contractive of the operator $\As$ and, in particular, it is upper bounded by
\begin{equation}\label{eq:ub}
    1-\frac{c}{b+1}.
\end{equation}
In general, there may be different pairs of $b,c$ for which $\As$ is enriched weakly contractive. Therefore, we solve a minimization problem whose objective function is the bound in~\eqref{eq:ub} to find the best possible pair $(b,c)$ such that enriched weak contractivity is guaranteed:
$$
\begin{aligned}
(b^\star,c^\star)=\argmin_{b\geq0, c\geq 0} \:\: & 
1-\frac{c}{b+1},\\
\text{s.t.} \quad &\: \abs{b I + A}\bone \leq (b-c+1)\bone,\\
& \: c\leq b+1.
\end{aligned}
$$
In our example, the solution to the above problem is ${b^\star = 0.695}$, $c^\star=0.29$, ensuring that the convergence rate is not greater than $0.83$ when $\theta=1/(b^\star+1)=0.59$. Instead, according to~\cite{davydov2024non}, selecting $\theta = \frac{1}{\text{diagL}(\Is-\As)}=0.48$ ensures a worse factor, namely $0.86$. Fig.~\ref{fig:impcrate}(left) shows the convergence of the distance between the iteration $\xb(k)$ and the fixed point $\xb^\star$ when the forward step method in \eqref{eq:fsm} $0.48$ is applied with the two different choices of $\theta$, corroborating the theoretical results by revealing an improved convergence rate.

\begin{figure}[!t]
    \centering
    \includegraphics[height=13em]{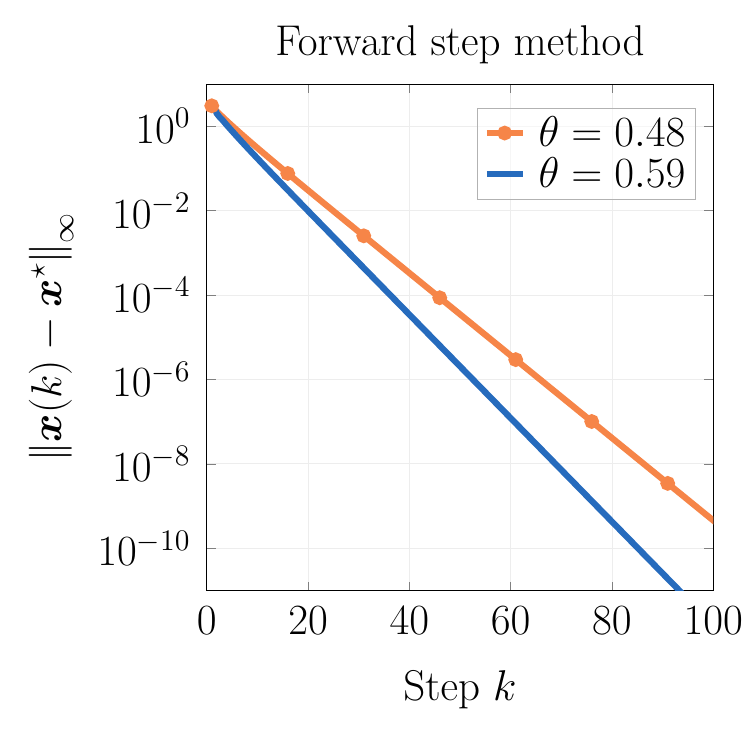}
    \includegraphics[height=13em]{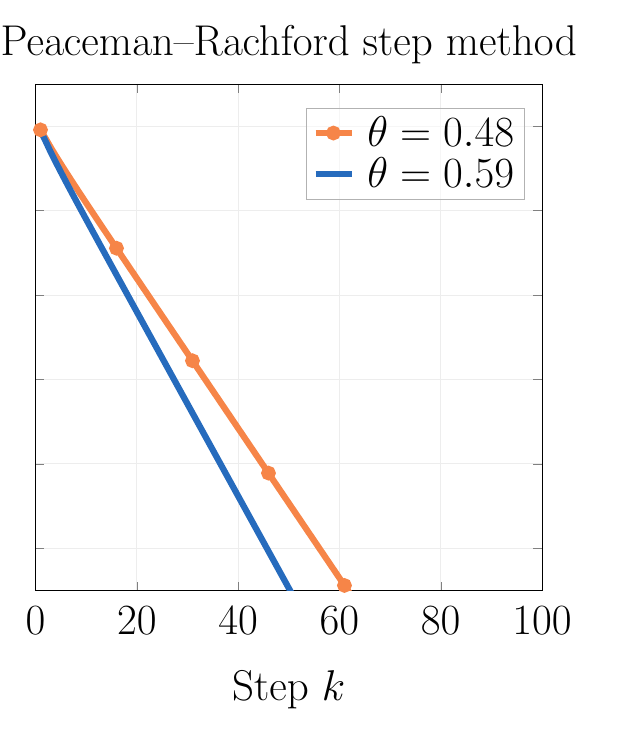}
    \caption{Residual versus number of iterations for the (left) forward-step method and (right) Peaceman-Rachford Splitting.}
    \label{fig:impcrate}
\end{figure}

We also compare the performance of the forward-backward and the Peaceman--Rachford step methods with the same choices of the step size, as their convergence can be guaranteed by retracing the results of~\cite{davydov2024non} in light of Theorem~\ref{thm:fsm}. The results for the forward-backward step method are comparable with those of the forward step method and thus are omitted, while the results for the Peaceman-Rachford step method are shown in Fig.~\ref{fig:impcrate}(right).

% \begin{table*}[!h]\label{tab:res}
% \begin{center}
%     \includegraphics[width=0.8\textwidth]{table_good.pdf}
%     \caption{Display of ratio $\theta^{\star}_{\textsc{mon}}/\theta^{\star}_{\textsc{ewc}}$ averaged over $10$ instances for different number of features $n\in\nat$: values lower than $1$ denote an improvement in the convergence rate obtained by exploiting $(b,c)$-enriched weak contractivity instead of monotonicity.}
%     \end{center}
% \end{table*}

\subsection{Composition of nonlinear diagonal operators with affine operators}\label{sec:dnl}
In this section, we consider the special class of \virg{\textit{nonlinear diagonal}} operators ${\Phi:\xb\mapsto[\phi_1(x_1),\cdots,\phi_n(x_n)]^\top}$ where $\phi_i:\rea\rightarrow\rea$ are assumed to be Lipschitz.
%
% We assume that each $\phi_i$ is Lipschitz with constant $\ell_{\phi}:=\Lip(\phi_i)$.
%
We consider operators $\Ts:=\Phi\circ\As$ resulting from the composition of a nonlinear diagonal operator with an affine operator.
This class of functions arises, for instance, in the context of training infinite-depth weight-tied neural network \cite{winston2020monotone,li2021training,Jafarpour21} by solving fixed-point problems of the form
\begin{equation}\label{eq:fpp}
\xb(k+1) = \Ts(\xb(k)) := \Phi(A\xb(k) + \bb),
\end{equation}
where ${\Phi:\rea^n\rightarrow \rea^n}$ plays the role of a nonlinear activation function while $A$, $\bb$ play the role of weight matrix and bias terms, ruling the continuous-time dynamics
$
{\dot{\xb}(t) = -\xb(t)+\Phi(A\xb+\bb)}.
$
We remark that, in the context of neural network, the vector $\bb$ is usually the result of two separate terms, namely the term related to the inputs $U\ub$ where $U\in\mathbb{R}^{n\times m}$ are the input-injection weights and $\ub\in\rea^m$ are the inputs, and the constant term related to the biases $\pb\in\rea^n$, yielding $\bb=U\ub+\pb$. 

Let us make the following technical assumption.
\begin{assum}\label{ass:actfun}
 The nonlinear diagonal operator ${\Phi(\xb)=[\phi_1(x_1),\cdots,\phi_n(x_n)]^\top}$ satisfies, for some $d_1\leq d_2$, the following condition
 $$
\frac{\phi_i(x)-\phi_i(x)}{x-y}\in[d_1,d_2],\qquad \forall x,y\in\rea,x\neq y.
 $$
\end{assum}

The most standard activation functions used in machine learning satisfy these bounds. For instance, the Leaky ReLu function
$$
\LReLU(x,\alpha) = \max\{\alpha x, x\},\quad \text{with}
\quad \alpha\in[0,1]
$$
satisfies Assumption~\ref{ass:actfun} with $d_1 = \alpha$ and $d_2=1$. This function has been introduced by Kaiming He \textit{et al.} in \cite{he2015delving} as an attempt to fix the \virg{dying ReLU} problem, i.e., the occurrence of dead neurons when using the ReLu function, which corresponds to the Leaky ReLu with $\alpha=0$. 

To test the performance resulting from the exploitation of enriched weak contractivity versus strong monotonicity, we randomly generate a matrix $M\in\rea^{n\times n}$ and a vector ${\bb\in\rea^n}$ following a Gaussian distribution with mean $\mu=0$ and variance $\sigma^2=1/n$. 
To ensure that $\Ts$ is enriched weakly contractive with respect to the diagonally weighted norm $\norm{\cdot}_{\infty,[\etab]^{-1}}$, we first derive a sufficient condition:
$$
\begin{aligned}
    \abs{b I + D\Ts(\xb)}\etab & = \abs{b I + D \Phi(A\xb(k) + \bb)}\etab \\
    %%%
    & \leq \max\left\{\abs{b I + d_1A}\etab,\abs{b I + d_2A}\etab\right\} \\
    %%%
    & \leq (b-c+1)\etab,
\end{aligned}
$$
Then, we compute $A$ by projecting $M$ such that the above condition is satisfied and the diagonal elements of $A$ are free:
$$
\begin{aligned}
\min_{b\geq c-1,A\in\rea^{n\times n}} \:\: & 
\norm{(A-M) - \diag(A-M)}_{F},\\
\text{s.t.} \qquad &\: \abs{b I + d_1A}\etab  \leq (b-c+1)\etab\\
& \: \abs{b I + d_2A}\etab  \leq (b-c+1)\etab
\end{aligned}
$$
for a given, fixed coefficient $c\geq 0$ and a vector $\etab\in\rea^{n}_{+}$, where $\norm{\cdot}_{F}$ denotes the Frobenius norm and $\diag(\cdot)$ returns a diagonal matrix whose elements are those in the diagonal of the input matrix.
%
%To give an idea of what kind of matrices can be generated by this method, we show next a matrix of dimension $n=5$ obtained by selecting $c=0.1$:
%
For instance, one matrix of dimension $n=5$ is
$$
A = \frac{1}{2}\begin{bmatrix}
    0.278&     0.111&    -0.280&     -0.134 & -0.098\\
    -0.189&    -0.739&    -0.207&     -0.105 & -0.066\\
    -0.408 & -0.355 & -0.203 & 0.301 & 0.039\\
    0.252 & 0.246 & -0.225 & -0.537 & -0.046\\
    0.144 & 0.253 & 0.288 & 0.225 & -0.395\\
\end{bmatrix},
$$
for which the operator $\Ts$ is $(b,c)$-enriched weakly contractive with $b=0.537$, $c=0.324$ and for ${\etab = [2.673, \: 1.181, \: 2.215, \: 1.261, \: 1.498]^\top}$.

In the experiments, we selected different choices for the parameter ${c=\{0.2,0.6,1,1.25,1.5,1.75,2\}}$ and we considered the Leaky ReLu activation function with $\alpha=0.1$. Also, we considered matrices $A\in\rea^{n\times n}$ with growing size ${n\in[5,200]}$ and, for each size $n$ and each parameter $c$, we run $10$ different experiments, computing the optimal step size $\theta$ both exploiting our Theorem~\ref{thm:nonlin-bene} and Theorem 26 in~\cite{davydov2024non}.
In particular, we denote by $\theta^{\star}_{\textsc{ewc}}$ the optimal step size obtained by minimizing the convergence rate $\rho^\star_{\textsc{ewc}}$ provided by Theorem~\ref{thm:nonlin-bene}, which exploits enriched weak contractivity, for any possible weighted norm $\norm{\cdot}_{\infty,[\etab]^{-1}}$:
$$
\rho^\star_{\textsc{ewc}} =1- c^\star_{\textsc{ewc}}\theta^\star_{\textsc{ewc}},\qquad \theta^{\star}_{\textsc{ewc}} = \frac{1}{b^\star_{\textsc{ewc}} +1}, 
$$
where $b^\star_{\textsc{ewc}},c^\star_{\textsc{ewc}},,\etab^\star_{\textsc{ewc}}$ are solutions of
$$
\begin{aligned}
\argmin_{b\geq0, c\geq 0,\etab>\bzero} \:\: & 
1-\frac{c}{b+1},\\
\text{s.t.} \qquad \quad &\:\abs{b I + d_1A}\etab  \leq (b-c+1)\etab\\
& \: \abs{b I + d_2A}\etab  \leq (b-c+1)\etab\\
& c\leq b+1.%\\
%%%
% & \: 1-c/(b+1)\geq 0.
\end{aligned}
$$

On the other hand, we denote by $\theta^{\star}_{\textsc{mon}}$ the optimal step size obtained by minimizing the convergence rate $\rho^\star_{\textsc{mon}}$ provided by~\cite[Theorem 26(ii)]{davydov2024non}, which exploits monotonicity, for any possible weighted norm $\norm{\cdot}_{\infty,[\etab]^{-1}}$, namely

$$
\rho^\star_{\textsc{mon}} =1- c^\star_{\textsc{mon}}\theta^\star_{\textsc{mon}},\qquad \theta^{\star}_{\textsc{mon}} = \frac{1}{1-\displaystyle \min_{i\in\{1,\cdots n\}}\{\alpha\cdot a_{ii},a_{ii}\}}, 
$$
where $c^\star_{\textsc{mon}},\etab^\star_{\textsc{mon}}$ are the solutions of (cfr.~\cite[Theorem 21(ii)]{davydov2024NN})
$$
\begin{aligned}
\argmin_{c\geq 0,\etab>\bzero} \:\: & 
1-c\theta^\star_{\textsc{mon}},\\
\text{s.t.} \qquad &\: \mzr{d_1A}\etab \leq (1-c) \etab\\
&\: \mzr{d_2A}\etab \leq (1-c) \etab\\
& c \leq  1/\theta^{\star}_{\textsc{mon}}%, \\
%%%
% & \: 1-c\theta^\star_{\textsc{mon}}\geq 0.
\end{aligned}
$$

\begin{figure}[!b]
    \centering
    \includegraphics[width=0.48\textwidth]{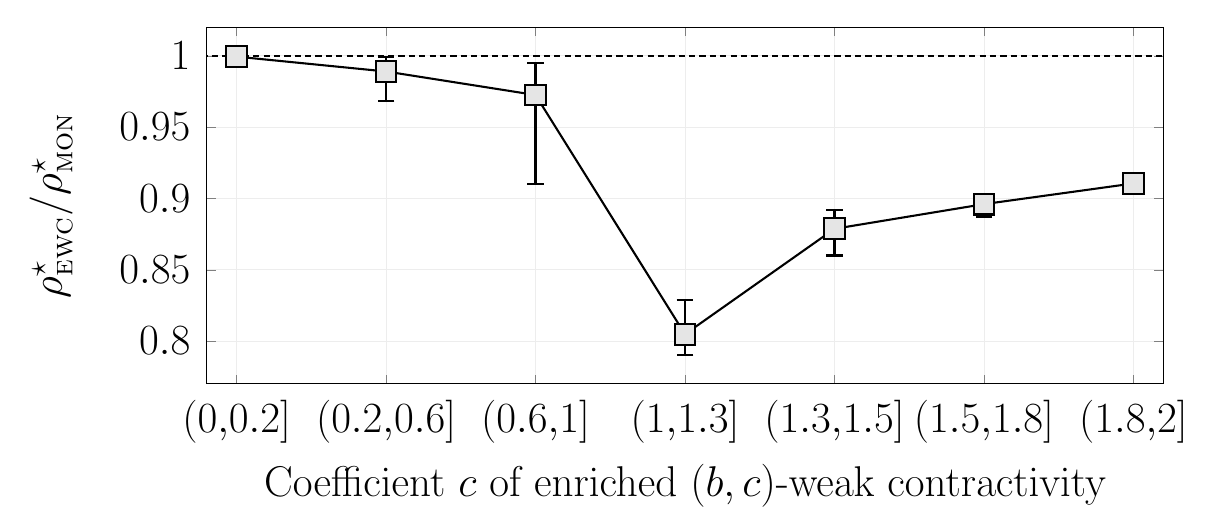}
    \caption{Display of the ratio $\rho^{\star}_{\textsc{ewc}}/\rho^{\star}_{\textsc{mon}}$ between the convergence rates averaged over $100$ instances for different number of features $n\in[5,200]$: values lower than $1$ denote an improvement in the convergence rate obtained by exploiting $(b,c)\text{-enriched}$ weak contractivity instead of $c$-strong monotonicity.}
    \label{fig:ratio}
\end{figure}

The results of the experiments are given in Fig~\ref{fig:ratio}, where we report the ratio between the two convergence rates $\rho^{\star}_{\textsc{mon}}$ and $\rho^{\star}_{\textsc{ewc}}$. Ratio values less than $1$ denote an improvement in the convergence rate bound thanks to the exploitation of enriched weak contractivity.
For small values of $c\approx0$, we notice that the ratio is very close to $1$ because the parameter $b$ does not influence significantly the convergence rate bound as it is close to $1$. 
Instead, for larger values of $c\approx 1$ we observe that the highest improvement  in the convergence rate bound is reached.
Finally, for large values of $c\gg 1$ this improvement is incrementally lost because in the limit of $c\rightarrow\infty$ it holds that $b\rightarrow\infty$, which implies that the convergence rate bound becomes again very close to $1$.

%
%Also, the standard logistic function
% %
% $$
% \sigma(x,\alpha)=\frac{\alpha}{1+e^{-x}},
% $$
% %
% satisfies Assumption~\ref{ass:actfun} with $d_1 = 0$ and $d_2=\alpha/4$. Finally, the hyperbolic tangent
% %
% $$
% \tanh(x) = \frac{e^x-e^{-x}}{e^x+e^{-x}}
% $$
% satisfies Assumption~\ref{ass:actfun} with $d_1 = 0$ and $d_2=1$.

% \section{Multi-Agent Systems}

% The importance of this result lies in its easier applicability than that of Theorem~\ref{thm:nonlin-bene}, especially in the context of multi-agent systems as discussed in the following example.

% \begin{exmp}\label{exmp:easier}

% \end{exmp}

\section{Application to Nonlinear Consensus\\ in Monotone Multi-Agent Systems}\label{sec:app2}
In this section, we consider multi-agent systems (MASs) composed of $n\in\nat$ agents modeled as dynamical systems with scalar state $x_i\in\mathbb{R}$, seeking consensus via discrete-time, distributed information exchange:
\begin{equation}\label{eq:agent}
x_i(k+1) = x_i(k)-\theta \sum_{i=1}^na_{ij}f_{ij}\Big(x_i(k)-x_j(k)\Big),
\end{equation}
where $f_{ij}:\mathbb{R}\rightarrow \mathbb{R}$ are Lipschitz nonlinear functions. We limit our attention to monotone discrete-time MASs in the sense of Definition~\ref{def:monosys}, i.e., systems whose dynamics is ruled by an order-preserving operator. For the dynamics in~\eqref{eq:agent}, order-preservation is guaranteed by $\partial f_{ij}/\partial x\geq 0$ a.e. according to Proposition~\ref{prop:op}. Monotone MASs are also called \textit{cooperative} especially in the continuous-time framework~\cite{Leenheer01,Hirsch03,Smith08,Bokharaie10,manfredi2017necessary}.

Coefficients $a_{ij}\in\{0,1\}$ are such that $a_{ij}=1$ if node $j$ can transmit information to node $i$, $a_{ij}=0$ if not, and $a_{ii}=0$. If $a_{ij}=1$, then agent $j$ is said to be a \emph{neighbor} of agent $i$, and we denote its set of neighbors by $\Nc_i=\left\{j\in \mathcal{V}: (i,j)\in \mathcal{E}\right\}$. The pattern of interconnections among the agents is given by a graph $\Gc=(\mathcal{V},\mathcal{E})$ where $\mathcal{V}=\{1,\ldots,n\}$ is the set of nodes representing the agents and $\mathcal{E}\subseteq \mathcal{V} \times \mathcal{V}$ is a set of directed edges. A directed edge $(i,j)\in \mathcal{E}$ exists if agent $i$ is influenced~by agent $j$, i.e., $a_{ij}=1$. The matrix $A=\{a_{ij}\}\in\mathbb{R}^{n\times n}$ is called the adjacency matrix.
%
%, while the linear operator $\Ls:\xb\mapsto(\diag(A\bone) -A)\xb$ is called the Laplacian operator. 
%
By defining 
$$
\Ls(\xb):=\diag(\xb)A-A\diag(\xb),\qquad g(A):=A\mapsto A\bone,
$$
we can rewrite the dynamics of the agents' network as follows
\begin{equation}\label{eq:mas}
\begin{aligned}
    \xb(k+1) &= \xb(k)-\theta (g\circ f\circ \Ls)(\xb(k)) \\
    & = (1-\theta)\xb(k)+\theta(\Is-g\circ f\circ \Ls)(\xb(k))
\end{aligned}
\end{equation}
where $f=[\cdots,f_{ij},\cdots]:\rea^{n\times n}\rightarrow\rea^{n \times n}$ and where $\circ$ denotes the composition operator.

The main result of this section is Theorem~\ref{th:consensus_DT}, which provides sufficient conditions on the local interaction rules $f_{ij}$ such that the operator ${\Ts:=\Is-g\circ f\circ\Ls}$ in \eqref{eq:agent} is enriched weakly contracting. Due to Theorem~\ref{thm:nonlin-bene}, enriched weak contractivity implies stability of the system and convergence of its state trajectories toward an equilibrium point, while accounting for heterogeneous local interaction rules.
This result is particularly interesting from a control perspective when addressing the problem of steering a MAS toward specific equilibrium points, by relying on partial and relative information, without the intervention of a central controller, as in formation control~\cite{oh2015survey} or distributed optimization~\cite{yang2019survey}. 
Moreover, we provide some extra sufficient conditions ensuring that the equilibrium point set coincides with the so-called consensus set ${\mathcal{C}=\{\alpha\bone: \alpha\in \rea \}}$.
The proposed sufficient condition is graph theoretical and based on the graph $\Gc$ describing the pattern of interconnections among the agents: It requires that there exists a globally reachable node in $\Gc$ and that consensus states are equilibrium points.

\begin{thm}\label{th:consensus_DT}
~
Consider a discrete-time MAS with agents dynamics as in~\eqref{eq:agent}. If the local interaction rules ${f_{ij}:\Xc\rightarrow\rea}$, with ${i= 1,\ldots,n}$, are Lipschitz with constant $L\geq 0$ and satisfy:
\begin{enumerate}[label=$(\roman*)$]
\item $\partial f_{ij}/\partial x\geq 0$ for almost every $x\in\rea$;
% \item $f_i(x+\alpha\bone) \leq f_i(x) + \alpha$ for any $\alpha\in\rea_{\geq 0}$;
\end{enumerate}
then the state trajectories globally, asymptotically converge  to one of its equilibrium points, if any, for all
\begin{equation}\label{eq:ss_mas}
    \theta\in\left(0,\frac{1}{\displaystyle L\max_{i=1,\cdots,n}\abs{\Nc_i}}\right).
\end{equation}
If it further holds that:
\begin{enumerate}[label=$(\roman*)$]
\setcounter{enumi}{1}
\item $f_{ij}(0)=0$ and $f_{ij}(x)\neq 0$ a.e. in a neighborhood of $0$;
\item the graph $\Gc$ has a globally reachable node;
\end{enumerate}
then the MAS converges asymptotically to a consensus state.
\end{thm}

\begin{proof}% of 
The state update of the MAS can be written in the form of the Krasnoselskij iteration~\eqref{eq:mas} on the operator ${\Ts:=\Is-f\circ\Ls}$, which is $(b,0)$-enriched weakly contracting for
$$
b\geq \text{diagL}(-\Ts)=\text{diagL}(g\circ f\circ\Ls)-1 = L\max_{i=1,\cdots,n}\abs{\Nc_i} -1
$$
w.r.t. $\norminf{\cdot}$ due to Theorem~\ref{thm:monoKSPC}, as it is is $\bone$-subhomogeneous (i.e., ${\Ts(\xb+\theta\bone) = \Ts(\xb) + \theta\bone}$) and its Jacobian matrix is Metzler by construction and by condition $(i)$.
When the system has at least one equilibrium point, one can thus exploit the result in Corollary~\ref{cor:nonlin-bene-mono} to establish that, for any initial condition and for all $\theta$ as in~\eqref{eq:ss_mas}, the state trajectories of the MAS converge to one of its equilibrium points, completing the first part of the proof.

While condition $(ii)$ implies that all consensus states are equilibrium points, condition $(iii)$ implies that there are no other equilibrium points other than the consensus states.
The graph $\Gc$ is aperiodic for any $\theta$ as in~\eqref{eq:ss_mas}, which ensures that the diagonal elements of the Jacobian matrix are strictly positive and, in turn, the presence of a self-loop at each node; moreover, the graph contains a globally reachable node due to condition $(iii)$.

By means of the definition of directional derivative, for every equilibrium point $\xb_e$, it holds
\begin{equation}\label{eq:dirder}
\begin{aligned}
D\Ts^{\pm}(c\bone)\xb_e &=
\lim_{h\rightarrow 0^{\pm}} \frac{\Ts(c\bone+h \xb_e)-\Ts(c\bone)}{h}\\
&=\lim_{h\rightarrow 0^{\pm}} \frac{c\bone+h\vb-c\bone}{h}
=\lim_{h\rightarrow 0^{\pm}} \frac{h\xb_e}{h}=\xb_e.
\end{aligned}
\end{equation}
By replacing $\xb_e=\bone$,~\eqref{eq:dirder} implies that the left/right Jacobian matrices $D\Ts^\pm$ are row-stochastic at any consensus point $c\bone$ with $c\in\rea$. By exploiting~\cite[Theorem 5.1]{Bullo18}, we conclude that $D\Ts^{\pm}$ have a simple unitary eigenvalue $\lambda = 1$ with corresponding eigenvector equal to $\bone$, unique up to a scaling factor.

Since the Krasnoselskij iteration is nonexpansive by Lemma~\ref{lem:lipcond}, then the set of equilibrium points is either empty or closed and convex by~\cite[Theorem 1]{Ferreira96}.
Now, if there is an equilibrium point $x_e\not\in\mathcal{C}$ that is not a consensus point, then all points ${c\bone + h x_e}$ with $h\in[0,1]$ and $c\in\mathbb{R}$ are also equilibrium points by construction, and thus
\begin{align*}
D\Ts^{\pm}(c\bone)x_e &=
\lim_{h\rightarrow 0^{\pm}} \frac{\Ts(c\bone+h x_e)-\Ts(c\bone)}{h}\\
&=\lim_{h\rightarrow 0^{\pm}} \frac{c\bone+hx_e-c\bone}{h}
=\lim_{h\rightarrow 0^{\pm}} \frac{hx_e}{h}=x_e,
\end{align*}

This means that $x_e$ is a second eigenvector (other than $v=\bone$) of the unitary eigenvalue $\lambda=1$ of $D\Ts^{\pm}$, which is a contradiction with respect to~\cite[Theorem 5.1]{Bullo18}. This proves that there do not exist equilibrium points $x_e$ outside the consensus set, completing the second part of the proof.
\end{proof}

Theorem \ref{th:consensus_DT} generalizes \cite[Theorem 6]{Deplano23} for multi-agent systems of the form \eqref{eq:agent} by accommodating non-continuously differentiable functions. This framework enables analysis of heterogeneous, asymmetric interactions such as:
$$
f_{ij}(x) = \max(\alpha x,x),\qquad f_{ij}(x) =  \min(\alpha x,x),
$$
with $\alpha >0$ ruling the asymmetry of the interaction. 
Notably, neighboring agents may also exhibit distinct interaction rules ($f_{ij} \neq f_{ji}$).
We remark that the above result can be generalized to deal with the special case $\alpha=0$, which leads to unilateral interactions in the sense of~\cite{manfredi2017necessary}, by considering more general notions of graphs, e.g., bicolored interaction graphs~\cite[Definition 10]{manfredi2017necessary}.

\section{Conclusions and Future Work}
The introduction of \textit{enriched weak contractivity} in this paper marks a significant step forward in the theory of monotone operators and monotone dynamical systems.
Indeed, enriched weak contractivity emerges as a powerful tool for generalizing foundational results, specifically by: (i) enabling the convergence of fixed-point iterations for monotone operators with larger step sizes in the Krasnoselskij iteration, thereby yielding improved bounds on convergence rates; and (ii) facilitating agreement in monotone multi-agent networks by accommodating Lipschitz, nonlinear, heterogeneous, and asymmetric interactions among agents.

Future work will investigate the role of enriched weak contractivity in other non-Euclidean spaces, e.g., those equipped with diagonally weighted $\ell_1$ norms. In fact, the iteration of operators may converge to a fixed point even though they are not contracting with respect to the $\ell_{\infty}$ norm (see Example~\ref{exmp:counter}) or the $\ell_2$ norm~\cite[Example 9]{bullo2021contraction}, while it is still possible that they are contracting with respect to $\ell_1$ norms. An interesting direction is the emergence of consensus in nonlinear multi-agent systems in the case of antagonistic interactions~\cite{altafini2012consensus}. 
\printbibliography

\end{document}